\documentclass[a4paper]{article}

\usepackage{xspace}
\usepackage[linesnumbered,ruled]{algorithm2e}
\usepackage{tikz}
\usepackage{amsmath}
\usepackage{amssymb}
\usepackage{graphicx}
\usepackage{hyperref}

\usepackage{fullpage}
\usepackage{amsthm}

\theoremstyle{plain}
\newtheorem{theorem}{Theorem}
\newtheorem{lemma}[theorem]{Lemma}
\newtheorem{claim}[theorem]{Claim}
\newtheorem{proposition}[theorem]{Proposition}
\newtheorem{corollary}[theorem]{Corollary}

\theoremstyle{definition}
\newtheorem{example}{Example}

\title{Weighted Triangle-free 2-matching Problem \\ with Edge-disjoint Forbidden Triangles}

\author{Yusuke Kobayashi\thanks{Research Institute for Mathematical Sciences, Kyoto University, Japan. Supported by JSPS KAKENHI Grant Numbers JP16K16010, 16H03118, and JP18H05291, Japan. Email: yusuke@kurims.kyoto-u.ac.jp}}

\begin{document}

\maketitle

\begin{abstract}
The weighted $\mathcal{T}$-free $2$-matching problem is the following problem: 
given an undirected graph $G$, a weight function on its edge set, and a set $\mathcal{T}$ of triangles in $G$, 
find a maximum weight $2$-matching containing no triangle in $\mathcal{T}$. 
When $\mathcal{T}$ is the set of all triangles in $G$, 
this problem is known as the weighted triangle-free $2$-matching problem, 
which is a long-standing open problem.  
A main contribution of this paper is to give a first polynomial-time algorithm for the 
weighted $\mathcal{T}$-free $2$-matching problem 
under the assumption that $\mathcal{T}$ is a set of edge-disjoint triangles.
In our algorithm, 
a key ingredient is to give an extended formulation representing the solution set, that is, 
we introduce new variables and represent 
the convex hull of the feasible solutions as a projection of another polytope in a higher dimensional space. 
Although our extended formulation has exponentially many inequalities, 
we show that the separation problem can be solved in polynomial time,
which leads to a polynomial-time algorithm for the weighted $\mathcal{T}$-free $2$-matching problem. 
\end{abstract}

\section{Introduction}\label{sec:intro}

\subsection{2-matchings without Short Cycles}

In an undirected graph, 
an edge set $M$ is said to be a {\em $2$-matching}\footnote{Although such an edge set is 
often called a {\em simple 2-matching} in the literature, we call it a {\em $2$-matching} to simplify the description.} 
if each vertex is incident to at most two edges in $M$.
Finding a $2$-matching of maximum size is a classical combinatorial optimization problem, which 
can be solved efficiently by using a matching algorithm.  
By imposing restrictions on $2$-matchings, 
various extensions have been introduced and studied in the literature.  
Among them,   
the problem of finding a maximum $2$-matching without short cycles has attracted attentions, 
because it has applications to approximation algorithms for TSP and its variants.  
We say that a $2$-matching $M$ is {\em $C_{\le k}$-free}
if $M$ contains no cycle of length $k$ or less, and 
the {\em $C_{\le k}$-free $2$-matching problem} is to find
a $C_{\le k}$-free $2$-matching of maximum size in a given graph. 
When $k \le 2$, every $2$-matching without self-loops and parallel edges is $C_{\le k}$-free, 
and hence the $C_{\le k}$-free $2$-matching problem can be solved in polynomial time. 
On the other hand, when $n/2 \le k \le n-1$, where $n$ is the number of vertices in the input graph, 
the $C_{\le k}$-free $2$-matching problem is NP-hard, because it decides the existence of a Hamiltonian cycle. 
These facts motivate us to investigate the borderline between polynomially solvable cases and NP-hard cases of the problem. 
Hartvigsen~\cite{HartD} gave
a polynomial-time algorithm for the $C_{\le 3}$-free $2$-matching problem, and
Papadimitriou showed that 
the problem is NP-hard when $k \ge 5$ (see \cite{CP80}). 
The polynomial solvability of the $C_{\le 4}$-free $2$-matching problem is still open, whereas 
some positive results are known for special cases. 
For the case when the input graph is restricted to be bipartite, 
Hartvigsen~\cite{Hart99}, Kir\'{a}ly~\cite{KiralyTR}, and Frank~\cite{Fra03} gave
min-max theorems, 
Hartvigsen~\cite{Hart06} and Pap~\cite{Pap07} designed polynomial-time algorithms, 
Babenko~\cite{Bab12} improved the running time, and 
Takazawa~\cite{TakDAM17} showed decomposition theorems. 
Recently, Takazawa~\cite{TakDO17,TakIPCO17} extended these results to a generalized problem. 
When the input graph is restricted to be subcubic, i.e., the maximum degree is at most three, 
B\'{e}rczi and V\'{e}gh~\cite{BV10} gave a polynomial-time algorithm for the $C_{\le 4}$-free $2$-matching problem. 
Relationship between $C_{\le k}$-free $2$-matchings and jump systems is studied in~\cite{BK0,Cun02,KST12}.

There are a lot of studies also on the weighted version of the $C_{\le k}$-free $2$-matching problem. 
In the weighted problem, an input consists of a graph and a weight function on the edge set, and 
the objective is to find a $C_{\le k}$-free $2$-matching of maximum total weight.  
Kir\'{a}ly proved
that the weighted $C_{\le 4}$-free $2$-matching problem
is NP-hard even if the input graph is restricted to be bipartite (see \cite{Fra03}), 
and a stronger NP-hardness result was shown in \cite{BK0}.
Under the assumption that the weight function satisfies a certain property called {\em vertex-induced on every square}, 
Makai~\cite{Mak07} gave a polyhedral description and 
Takazawa~\cite{Tak0} designed a combinatorial polynomial-time algorithm 
for the weighted $C_{\le 4}$-free $2$-matching problem in bipartite graphs. 
The case of $k=3$, which we call the {\em weighted triangle-free $2$-matching problem}, is a long-standing open problem. 
For the weighted triangle-free $2$-matching problem in subcubic graphs, 
Hartvigsen and Li~\cite{HLIPCO} gave a polyhedral description and 
a polynomial-time algorithm, followed by a slight generalized polyhedral description by B\'{e}rczi~\cite{Berczi0} 
and another polynomial-time algorithm by Kobayashi~\cite{Kob0}. 
Relationship between $C_{\le k}$-free $2$-matchings and discrete convexity is studied in~\cite{Kob0,KobDAM14,KST12}.

\subsection{Our Results}

The previous papers on the weighted triangle-free $2$-matching problem~\cite{Berczi0,HLIPCO,Kob0} 
deal with a generalized problem in which 
we are given a set $\mathcal{T}$ of forbidden triangles as an input in addition to a graph and a weight function. 
The objective is to find a maximum weight $2$-matching that contains no triangle in $\mathcal{T}$, 
which we call the {\em weighted $\mathcal{T}$-free $2$-matching problem}. 
In this paper, we focus on the case when $\mathcal{T}$ is a set of edge-disjoint triangles, i.e., 
no pair of triangles in $\mathcal{T}$ shares an edge in common.  
A main contribution of this paper is to give a first polynomial-time algorithm for the 
weighted $\mathcal{T}$-free $2$-matching problem 
under the assumption that $\mathcal{T}$ is a set of edge-disjoint triangles.
Note that we impose an assumption only on $\mathcal{T}$, and 
no restriction is required for the input graph. 
We now describe the formal statement of our result.

Let $G=(V, E)$ be an undirected graph with vertex set $V$ and edge set $E$, 
which might have self-loops and parallel edges. 
For a vertex set $X \subseteq V$, 
let $\delta_G(X)$ denote the set of edges between $X$ and $V \setminus X$. 
For $v \in V$, $\delta_G(\{v\})$ is simply denoted by $\delta_G(v)$.  
For $v \in V$, 
let $\dot\delta_G (v)$ denote the multiset of edges incident to $v\in V$, 
that is, a self-loop incident to $v$ is counted twice. 
We omit the subscript $G$ if no confusion may arise.  
For $b \in \mathbf{Z}_{\ge 0}^V$, an edge set $M \subseteq E$ is said to be a {\em $b$-matching} (resp.~{\em $b$-factor}) if
$|M \cap \dot\delta(v)| \le b(v)$ (resp.~$|M \cap \dot\delta(v)| = b(v)$) for every $v \in V$.
If $b(v)=2$ for every $v \in V$, a $b$-matching and a $b$-factor are called 
a {\em $2$-matching} and a {\em $2$-factor}, respectively. 
Let $\mathcal{T}$ be a set of triangles in $G$, where a {\em triangle} is a cycle of length three. 
For a triangle $T$, let $V(T)$ and $E(T)$ denote the vertex set and the edge set of $T$, respectively. 
An edge set $M \subseteq E$ is said to be {\em $\mathcal{T}$-free} 
 if $E(T) \not\subseteq M$ for every $T \in \mathcal{T}$. 
For a vertex set $S \subseteq V$, 
let $E[S]$ denote the set of all edges with both endpoints in $S$. 
For an edge weight vector $w \in \mathbf{R}^E$, 
we consider the problem of finding a $\mathcal{T}$-free $b$-matching (resp.~$b$-factor)
maximizing $w(M)$, 
which we call the {\em weighted $\mathcal{T}$-free $b$-matching (resp.~$b$-factor) problem}. 
Note that, for a set $A$ and a vector $c \in \mathbf{R}^A$, 
we denote $c(A) = \sum_{a \in A} c(a)$.

Our main result is formally stated as follows. 

\begin{theorem}\label{thm:algo}
There exists a polynomial-time algorithm for the following problem: 
given a graph $G=(V, E)$, $b(v) \in \mathbf{Z}_{\ge 0}$ for each $v \in V$,  
a set $\mathcal{T}$ of edge-disjoint triangles, and a weight $w(e) \in \mathbf{R}$ for each $e \in E$, 
find a $\mathcal{T}$-free $b$-factor $M \subseteq E$ that maximizes the total weight $w(M)$.
\end{theorem}

A proof of this theorem is given in Section~\ref{sec:algo}. 
Since finding a maximum weight $\mathcal{T}$-free $b$-matching can be reduced to 
finding a maximum weight  $\mathcal{T}$-free $b$-factor 
by adding dummy vertices and zero-weight edges, 
Theorem~\ref{thm:algo} implies the following corollary. 

\begin{corollary}\label{cor:algo1}
There exists a polynomial-time algorithm for the following problem: 
given a graph $G=(V, E)$, $b(v) \in \mathbf{Z}_{\ge 0}$ for each $v \in V$,  
a set $\mathcal{T}$ of edge-disjoint triangles, and a weight $w(e) \in \mathbf{R}$ for each $e \in E$, 
find a $\mathcal{T}$-free $b$-matching $M \subseteq E$ that maximizes the total weight $w(M)$.
\end{corollary}

In particular, 
we can find 
a $\mathcal{T}$-free $2$-matching (or $2$-factor) $M \subseteq E$ that maximizes the total weight $w(M)$ in polynomial time
if $\mathcal{T}$ is a set of edge-disjoint triangles.

\subsection{Key Ingredient: Extended Formulation}
\label{sec:ef}

A natural strategy to solve the maximum weight $\mathcal{T}$-free $b$-factor problem is 
to give a polyhedral description of the $\mathcal{T}$-free $b$-factor polytope
as Hartvigsen and Li~\cite{HLIPCO} did for the subcubic case. 
However, as we will see in Example~\ref{exam:01}, 
giving a system of inequalities that represents the $\mathcal{T}$-free $b$-factor polytope seems to be quite difficult
even when $\mathcal{T}$ is a set of edge-disjoint triangles. 
A key idea of this paper is to give an extended formulation of the $\mathcal{T}$-free $b$-factor polytope, that is, 
we introduce new variables and represent 
the $\mathcal{T}$-free $b$-factor polytope as a projection of another polytope in a higher dimensional space. 

Extended formulations of polytopes arising from various combinatorial optimization problems have been intensively studied in the literature, and  
the main focus in this area is on
the number of inequalities that are required to represent the polytope. 
If a polytope has an extended formulation with polynomially many inequalities, then 
we can optimize a linear function in the original polytope by the ellipsoid method (see e.g.~\cite{GLSbook}). 
On the other hand, even if a linear function on a polytope can be optimized in polynomial time, 
the polytope does not necessarily have an extended formulation of polynomial size. 
In this context, the existence of a polynomial size extended formulation has been attracted attentions. 
See survey papers~\cite{CCZ10,Kai11} for previous work on extended formulations. 

In this paper, under the assumption that $\mathcal{T}$ is a set of edge-disjoint triangles, 
we give an extended formulation of the $\mathcal{T}$-free $b$-factor polytope 
that has exponentially many inequalities (Theorem~\ref{thm:main01}). 
In addition, we show that the separation problem for the extended formulation is solvable in polynomial time, and hence
we can optimize a linear function on the $\mathcal{T}$-free $b$-factor polytope by the ellipsoid method in polynomial time. 
This yields a first polynomial-time algorithm for the weighted $\mathcal{T}$-free $b$-factor (or $b$-matching) problem.  
Note that it is rare that the first polynomial-time algorithm was designed
with the aid of an extended formulation. 
To the best of our knowledge, 
the {\em weighted linear matroid parity problem}  
was the only such problem before this paper (see~\cite{IK2017}).

\subsection{Organization of the Paper}
The rest of this paper is organized as follows. 
In Section~\ref{sec:EF}, we introduce an extended formulation of the $\mathcal{T}$-free $b$-factor polytope, 
whose correctness proof is given in Section~\ref{sec:proof}. 
In Section~\ref{sec:propertyQ}, we show a few technical lemmas that will be used in the proof. 
In Section~\ref{sec:algo}, we give a polynomial-time algorithm for the weighted $\mathcal{T}$-factor problem and prove Theorem~\ref{thm:algo}. 
Finally, we conclude this paper with remarks in Section~\ref{sec:conclusion}. 
Some of the proofs are postponed to the appendix.


\section{Extended Formulation of the $\mathcal{T}$-free $b$-factor Polytope}
\label{sec:EF}

Let $G=(V, E)$ be a graph, $b \in \mathbf{Z}^V_{\ge 0}$ be a vector, and $\mathcal T$ be a set of forbidden triangles. 
Throughout this paper, we only consider the case when triangles in $\mathcal T$ are mutually edge-disjoint. 

For an edge set $M \subseteq E$, 
define its characteristic vector $x_M \in \mathbf{R}^E$ by 
\begin{equation}
x_M(e)=
\begin{cases}
1 & \mbox{if $e \in M$,} \\
0 & \mbox{otherwise}. 
\end{cases} \label{eq:chara}
\end{equation}
The {\em $\mathcal{T}$-free $b$-factor polytope} is defined as 
${\rm conv} \{x_M \mid \mbox{$M$ is a $\mathcal{T}$-free $b$-factor in $G$} \}$, 
where ${\rm conv}$ denotes the convex hull of vectors, 
and the {\em $b$-factor polytope} is defined similarly. 
Edmonds~\cite{Edm65} shows that the $b$-factor polytope is determined by the following inequalities. 
\begin{align}
&x(\dot\delta(v)) = b(v) & &(v \in V) \label{eq:01}\\
&0 \le x(e) \le 1 & & (e \in E) \label{eq:02} \\
&\sum_{e \in F_0} x(e) + \sum_{e \in F_1} (1- x(e)) \ge 1  & & ((S, F_0, F_1) \in \mathcal{F}) \label{eq:03}
\end{align}
Here, $\mathcal{F}$ is the set of all triples $(S, F_0, F_1)$ such that 
$S \subseteq V$, $(F_0, F_1)$ is a partition of $\delta(S)$, and 
$b(S) + |F_1|$ is odd. 
Note that 
$x(\dot\delta(v)) = \sum_{e \in \dot\delta(v)} x(e)$ and
$x(e)$ is added twice if $e$ is a self-loop incident to $v$.

In order to deal with $\mathcal{T}$-free $b$-factors, 
we consider the following constraint in addition to (\ref{eq:01})--(\ref{eq:03}). 
\begin{align}
&x(E(T)) \le 2 & & (T \in \mathcal{T}) \label{eq:triangle}
\end{align}
However, as we will see in Example~\ref{exam:01}, 
the system of inequalities (\ref{eq:01})--(\ref{eq:triangle}) does not 
represent the $\mathcal{T}$-free $b$-factor polytope. 
Note that when we consider uncapacitated $2$-factors, i.e., 
we are allowed to use two copies of the same edge, 
it is shown by Cornuejols and Pulleyblank~\cite{CP80uncapacitated} that 
the $\mathcal{T}$-free uncapacitated $2$-factor polytope is represented by 
$x(e) \ge 0$ for $e \in E$, $x(\dot \delta(v)) = 2$ for $v \in V$, and (\ref{eq:triangle}).

\begin{example}\label{exam:01}
Consider the graph $G=(V, E)$ in Figure~\ref{fig:01}. 
Let $b(v)=2$ for every $v \in V$ and $\mathcal{T}$ be the set of all triangles in $G$. 
Then, $G$ has no $\mathcal{T}$-free $b$-factor, i.e., the $\mathcal{T}$-free $b$-factor polytope is empty. 
For $e \in E$, let $x(e)=1$ if $e$ is drawn as a blue line in Figure~\ref{fig:01}
and let $x(e)=\frac{1}{2}$ otherwise. 
Then, we can easily check that $x$ satisfies (\ref{eq:01}), (\ref{eq:02}), and (\ref{eq:triangle}). 
Furthermore, since $x$ is represented as a linear combination of two $b$-factors $M_1$ and $M_2$ shown in Figures~\ref{fig:02} and~\ref{fig:03}, 
$x$ satisfies (\ref{eq:03}). 
\end{example}

\begin{figure}
 \begin{minipage}{0.32\hsize}
\begin{center}
\includegraphics[clip,width=3.8cm]{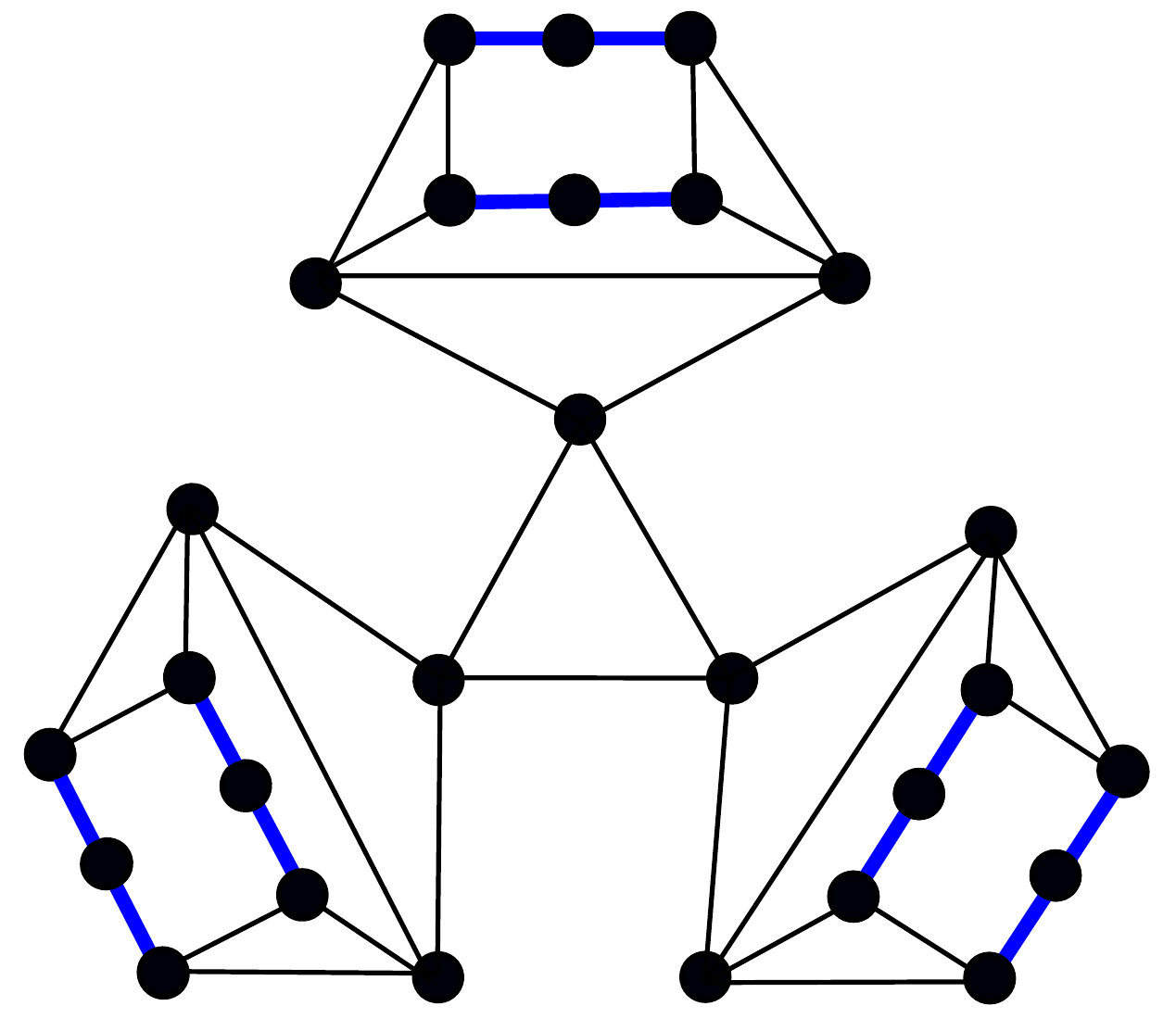}
\caption{Graph $G=(V, E)$}
\label{fig:01}
\end{center}
 \end{minipage}
 \begin{minipage}{0.32\hsize}
\begin{center}
\includegraphics[clip,width=3.8cm]{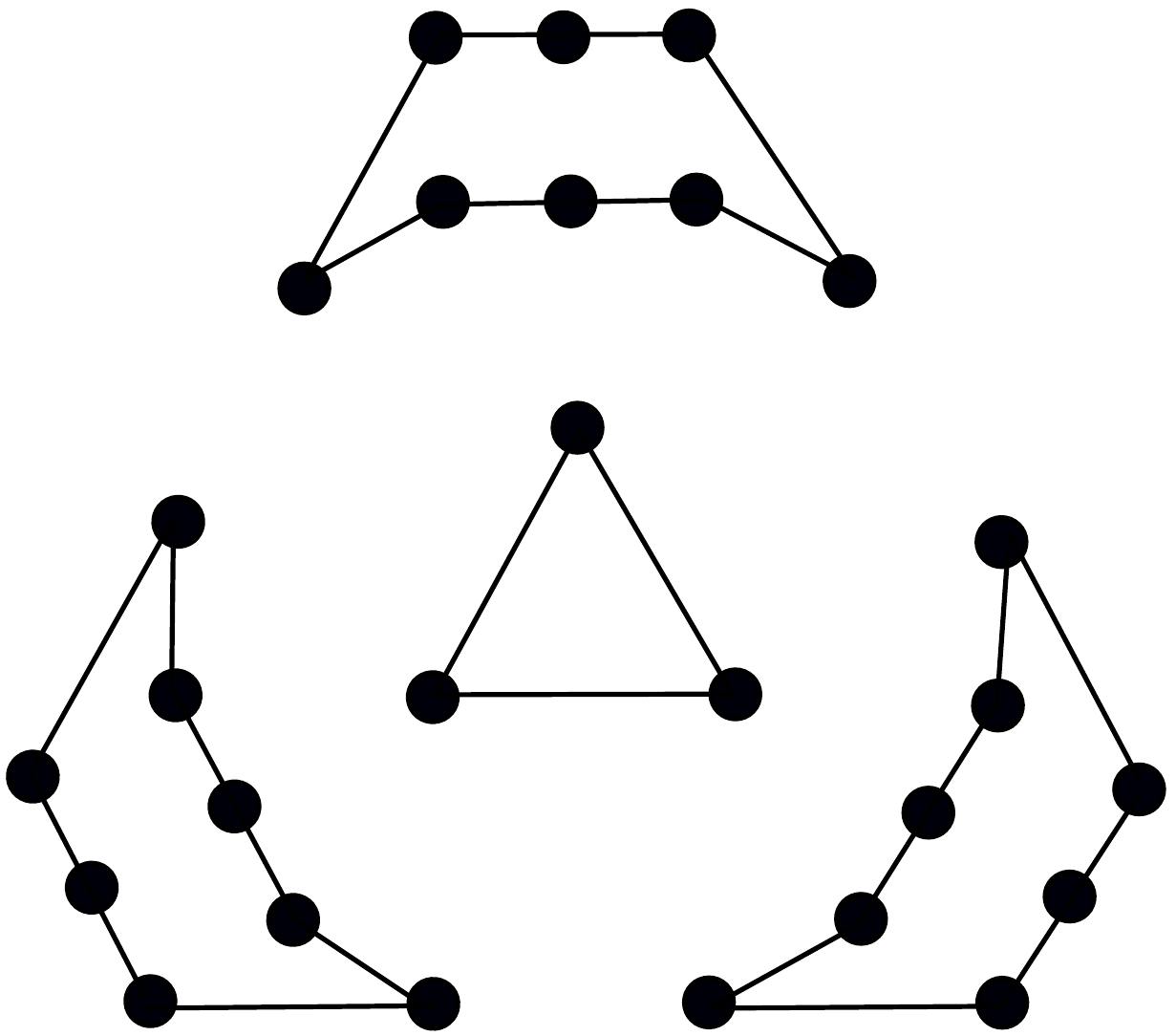}
\caption{$b$-factor $M_1$}
\label{fig:02}
\end{center}
 \end{minipage}
 \begin{minipage}{0.32\hsize}
\begin{center}
\includegraphics[clip,width=3.8cm]{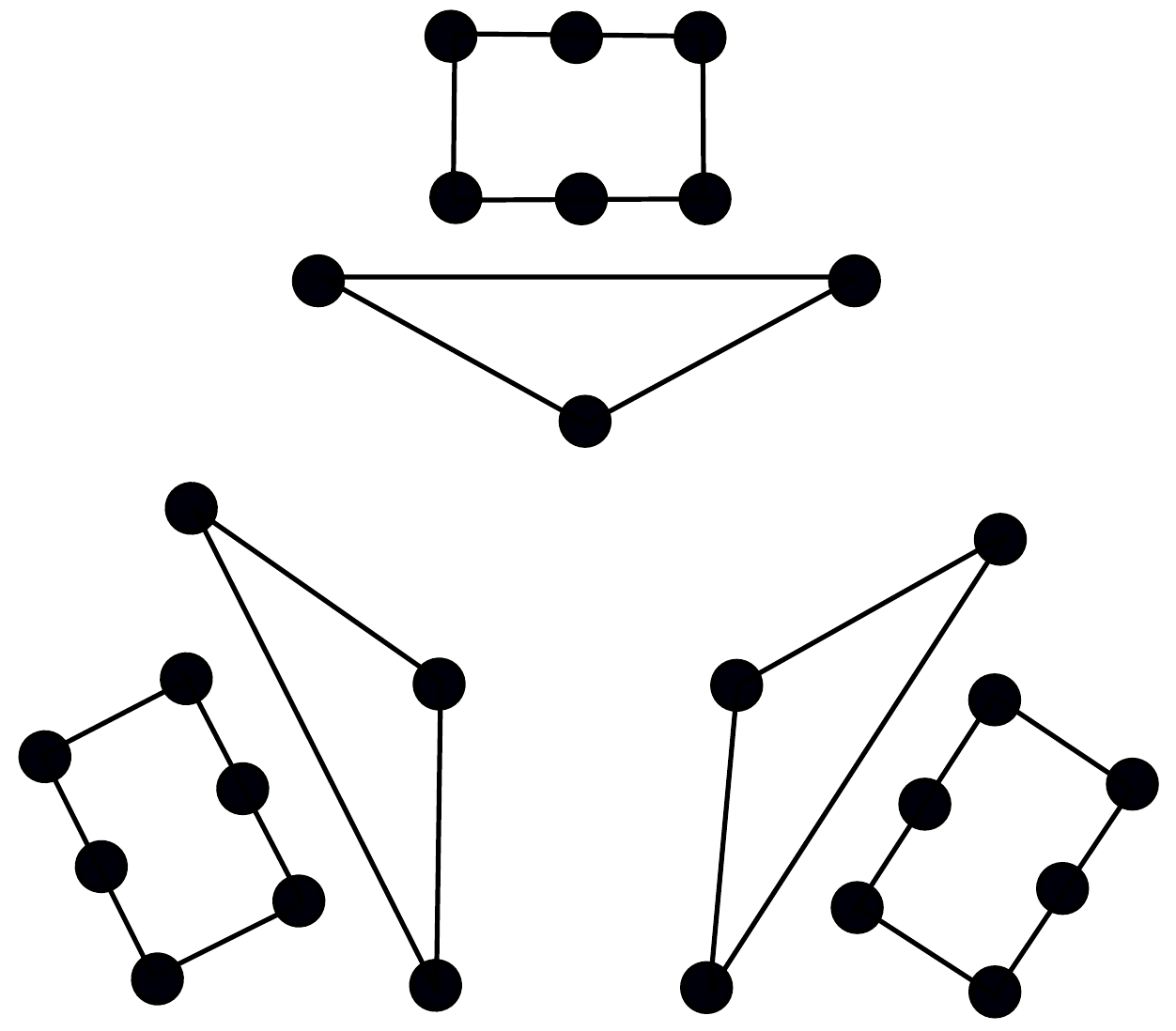}
\caption{$b$-factor $M_2$}
\label{fig:03}
\end{center}
\end{minipage}
\end{figure}

In what follows in this section, we introduce new variables and 
give an extended formulation of the $\mathcal{T}$-free $b$-factor polytope.
For $T \in \mathcal{T}$, 
we denote $\mathcal{E}_T = \{ J \subseteq E(T) \mid J \not = E(T)\}$. 
For $T \in \mathcal{T}$
 and $J \in \mathcal{E}_T$, 
we introduce a new variable $y(T, J)$. 
Roughly, $y(T, J)$ denotes the fraction of $b$-factors $M$ satisfying $M \cap E(T) = J$. 
In particular, when $x$ and $y$ are integral, 
$y(T, J) = 1$ if and only if the $b$-factor $M$ corresponding to $(x, y)$ satisfies $M \cap E(T) = J$.
We consider the following inequalities. 
\begin{align}
&\sum_{J \in \mathcal{E}_T} y(T, J) = 1 & &(T \in \mathcal{T}) \label{eq:04} \\
&\sum_{e \in J \in \mathcal{E}_T} y(T, J) = x(e)                   & & (T \in \mathcal{T},\ e \in E(T)) \label{eq:05} \\
&y(T, J) \ge 0 & & (T \in \mathcal{T},\ J \in \mathcal{E}_T) \label{eq:06}
\end{align}
If $T$ is clear from the context, $y(T, J)$ is simply denoted by $y(J)$.  
Since triangles in $\mathcal{T}$ are edge-disjoint, this causes no ambiguity unless $J=\emptyset$. 
In addition, for $\alpha, \beta \in E(T)$, $y(\{\alpha\})$, $y(\{\alpha, \beta\})$, and $y(\emptyset)$ are 
simply denoted by $y_{\alpha}$, $y_{\alpha \beta}$, and $y_{\emptyset}$, respectively.

We now strengthen (\ref{eq:03}) by using $y$. 
For $(S, F_0, F_1) \in \mathcal{F}$, 
let $\mathcal{T}_S = \{T \in \mathcal{T} \mid E(T) \cap \delta(S) \not= \emptyset \}$. 
For $T \in \mathcal{T}_S$ with
$E(T)=\{\alpha, \beta, \gamma\}$ and $E(T) \cap \delta(S) = \{\alpha, \beta\}$, 
we define 
$$
q^*(T) =
\begin{cases}
y_{\alpha} + y_{\alpha \gamma} & \mbox{if $\alpha \in F_0$ and $\beta \in F_1$,} \\
y_{\beta} + y_{\beta \gamma} & \mbox{if $\beta \in F_0$ and $\alpha \in F_1$,} \\
y_\emptyset + y_{\gamma}  & \mbox{if $\alpha, \beta \in F_1$, } \\
y_{\alpha \beta} & \mbox{if $\alpha, \beta \in F_0$.}
\end{cases}
$$ 
Note that this value depends on $(S, F_0, F_1) \in \mathcal{F}$ and $y$, but it is simply denoted by $q^*(T)$ for a notational convenience. 
We consider the following inequality. 
\begin{align}
\sum_{e \in F_0} x(e) + \sum_{e \in F_1} (1- x(e)) - \sum_{T \in \mathcal{T}_S} 2 q^*(T) \ge 1 
   \qquad ((S, F_0, F_1) \in \mathcal{F})   \label{eq:07} 
\end{align}
For $T \in \mathcal{T}_S$ with
$E(T)=\{\alpha, \beta, \gamma\}$ and $E(T) \cap \delta(S) = \{\alpha, \beta\}$, 
the contribution of $\alpha, \beta$, and $T$ to the left-hand side of (\ref{eq:07}) 
is equal to the fraction of $b$-factors $M$ such that $|M \cap \{ \alpha, \beta \}| \not\equiv |F_1 \cap \{\alpha, \beta\}| \pmod 2$ 
by the following observations. 
\begin{itemize}
\item
If $\alpha \in F_0$ and $\beta \in F_1$, then (\ref{eq:04}) and (\ref{eq:05}) show that 
$x(\alpha) = y_{\alpha} + y_{\alpha \beta} + y_{\alpha \gamma}$ and 
$1-x(\beta) = 1 - (y_{\beta} + y_{\alpha \beta} + y_{\beta \gamma}) = y_{\emptyset} + y_{\alpha} + y_{\gamma} + y_{\alpha \gamma}$.
Therefore, 
$x(\alpha) + (1-x(\beta)) - 2 q^*(T) = y_\emptyset + y_{\gamma} + y_{\alpha \beta}$, 
which denotes the fraction of $b$-factors $M$ such that $|M \cap \{ \alpha, \beta \}|$ is even.

\item
If $\beta \in F_0$ and $\alpha \in F_1$, then (\ref{eq:04}) and (\ref{eq:05}) show that
$(1-x(\alpha)) + x(\beta) - 2 q^*(T) = y_\emptyset + y_{\gamma} + y_{\alpha \beta}$, 
which denotes the fraction of $b$-factors $M$ such that $|M \cap \{ \alpha, \beta \}|$ is even.

\item
If $\alpha, \beta \in F_1$, then (\ref{eq:04}) and (\ref{eq:05}) show that
$(1-x(\alpha)) + (1-x(\beta)) - 2 q^*(T) = y_\alpha + y_\beta + y_{\alpha \gamma} + y_{\beta \gamma}$, 
which denotes the fraction of $b$-factors $M$ such that $|M \cap \{ \alpha, \beta \}|$ is odd.

\item
If $\alpha, \beta \in F_0$, then (\ref{eq:04}) and (\ref{eq:05}) show that
$x(\alpha) + x(\beta) - 2 q^*(T) = y_\alpha + y_\beta + y_{\alpha \gamma} + y_{\beta \gamma}$, 
which denotes the fraction of $b$-factors $M$ such that $|M \cap \{ \alpha, \beta \}|$ is odd.
\end{itemize}

Let $P$ be the polytope defined by 
$$
P = \{ (x, y) \in \mathbf{R}^E \times \mathbf{R}^Y \mid \mbox{$x$ and $y$ satisfy (\ref{eq:01}), (\ref{eq:02}), and (\ref{eq:triangle})--(\ref{eq:07})} \}, 
$$
where $Y = \{ (T, F) \mid T \in \mathcal{T},\ F \in \mathcal{E}_T\}$. 
Note that we do not need (\ref{eq:03}), because it is implied by (\ref{eq:07}). 
Define the projection of $P$ onto $E$ as 
$$
{\rm proj}_E (P) = \{ x \in \mathbf{R}^E \mid \mbox{There exists $y \in \mathbf{R}^Y$ such that $(x, y) \in P$} \}. 
$$
Our aim is to show that   
${\rm proj}_E (P)$ is equal to the $\mathcal{T}$-free $b$-factor polytope. 
It is not difficult to see that the $\mathcal{T}$-free $b$-factor polytope is contained in ${\rm proj}_E (P)$. 

\begin{lemma}\label{lem:easydir}
The $\mathcal{T}$-free $b$-factor polytope is contained in ${\rm proj}_E (P)$. 
\end{lemma}

\begin{proof}
Suppose that $M \subseteq E$ is a $\mathcal{T}$-free $b$-factor in $G$ and
define $x_M \in \mathbf{R}^E$ by (\ref{eq:chara}). 
For $T \in \mathcal T$ and $J \in \mathcal{E}_T$, 
define 
$$
y_M(T, J)=
\begin{cases}
1 & \mbox{if $M \cap E(T) = J$,} \\
0 & \mbox{otherwise}. 
\end{cases}
$$
We can easily see that 
$(x_M, y_M)$ satisfies (\ref{eq:01}), (\ref{eq:02}), and (\ref{eq:triangle})--(\ref{eq:06}).
Thus, it suffices to show that $(x_M, y_M)$ satisfies (\ref{eq:07}). 
Assume to the contrary that (\ref{eq:07}) does not hold for $(S, F_0, F_1) \in \mathcal{F}$. 
Then, $x_M(e) = 0$ for every $e \in F_0 \setminus \bigcup_{T \in \mathcal{T}_S} E(T)$ and 
$x_M(e)=1$ for every $e \in F_1 \setminus \bigcup_{T \in \mathcal{T}_S} E(T)$. 
Furthermore, since 
the contribution of $E(T) \cap \delta(S)$ and $T$ to the left-hand side of (\ref{eq:07}) 
is equal to $1$ if and only if 
$|M \cap E(T) \cap \delta(S)| \not\equiv |F_1 \cap E(T)| \pmod 2$, 
we obtain 
$|M \cap E(T) \cap \delta(S)| \equiv |F_1 \cap E(T)| \pmod 2$ for every $T \in \mathcal{T}_S$. 
Then, 
\begin{align*}
|M \cap \delta(S)| 
&= |(M \cap \delta(S)) \setminus \bigcup_{T \in \mathcal{T}_S} E(T)| + \sum_{T \in \mathcal{T}_S} |M \cap E(T) \cap \delta(S)| \\
&\equiv |F_1 \setminus \bigcup_{T \in \mathcal{T}_S} E(T)| + \sum_{T \in \mathcal{T}_S} |F_1 \cap E(T)| = |F_1|.  
\end{align*}
Since $M$ is a $b$-factor, 
it holds that $|M \cap \delta(S)| \equiv b(S) \pmod{2}$, 
which contradicts that $b(S) +  |F_1|$ is odd. 
\end{proof}

To prove the opposite inclusion (i.e., ${\rm proj}_E (P)$ is contained in the $\mathcal{T}$-free $b$-factor polytope),  
we consider a relaxation of (\ref{eq:07}). 
For $T \in \mathcal{T}_S$ with
$E(T)=\{\alpha, \beta, \gamma\}$ and $E(T) \cap \delta(S) = \{\alpha, \beta\}$, 
we define 
$$
q(T) =
\begin{cases}
y_{\alpha} + y_{\alpha \gamma} & \mbox{if $\alpha \in F_0$ and $\beta \in F_1$,} \\
y_{\beta} + y_{\beta \gamma} & \mbox{if $\beta \in F_0$ and $\alpha \in F_1$,} \\
y_{\gamma}  & \mbox{if $\alpha, \beta \in F_1$, } \\
0                & \mbox{if $\alpha, \beta \in F_0$.}
\end{cases}
$$ 
Since $q(T) \le q^*(T)$ for every $T \in \mathcal{T}_S$, 
the following inequality is a relaxation of (\ref{eq:07}). 
\begin{align}
\sum_{e \in F_0} x(e) + \sum_{e \in F_1} (1- x(e)) - \sum_{T \in \mathcal{T}_S} 2 q(T) \ge 1 
   \qquad ((S, F_0, F_1) \in \mathcal{F})   \label{eq:100} 
\end{align}
Note that there is a difference between (\ref{eq:07}) and (\ref{eq:100}) in the following cases. 
\begin{itemize}
\item
If $\alpha, \beta \in F_1$, then the contribution of $\alpha, \beta$, and $T$ to the left-hand side of (\ref{eq:100}) is
$(1-x(\alpha)) + (1-x(\beta)) - 2 q(T) = y_\alpha + y_\beta + y_{\alpha \gamma} + y_{\beta \gamma} + 2 y_\emptyset$. 
\item
If $\alpha, \beta \in F_0$, then the contribution of $\alpha, \beta$, and $T$ to the left-hand side of (\ref{eq:100}) is
$x(\alpha) + x(\beta) - 2 q(T) = y_\alpha + y_\beta + y_{\alpha \gamma} + y_{\beta \gamma} + 2 y_{\alpha \beta}$. 
\end{itemize}
Define a polytope $Q$ and its projection onto $E$ as 
\begin{align*}
& Q = \{ (x, y) \in \mathbf{R}^E \times \mathbf{R}^Y  \mid \mbox{$x$ and $y$ satisfy (\ref{eq:01}), (\ref{eq:02}), (\ref{eq:triangle})--(\ref{eq:06}), and (\ref{eq:100})} \}, \\
& {\rm proj}_E(Q) = \{ x \in \mathbf{R}^E \mid \mbox{There exists $y \in \mathbf{R}^Y$  such that $(x, y) \in Q$} \}. 
\end{align*}
Since (\ref{eq:100}) is implied by (\ref{eq:07}), 
we have that $P \subseteq Q$ and ${\rm proj}_E(P) \subseteq {\rm proj}_E(Q)$. 
In what follows in Sections~\ref{sec:propertyQ} and~\ref{sec:proof}, 
we show the following proposition. 

\begin{proposition}\label{prop:main}
${\rm proj}_E(Q)$ is contained in the $\mathcal{T}$-free $b$-factor polytope. 
\end{proposition}

By Lemma~\ref{lem:easydir}, Proposition~\ref{prop:main}, and ${\rm proj}_E(P) \subseteq {\rm proj}_E(Q)$, 
we obtain the following theorem.

\begin{theorem}\label{thm:main01}
Let $G=(V, E)$ be a graph, $b(v) \in \mathbf{Z}_{\ge 0}$ for each $v \in V$, and 
let $\mathcal{T}$ be a set of edge-disjoint triangles. 
Then, both ${\rm proj}_E(P)$ and ${\rm proj}_E(Q)$ are equal to 
the $\mathcal{T}$-free $b$-factor polytope. 
\end{theorem}

We remark here that we do not know how to prove directly that ${\rm proj}_E(P)$ is contained in the $\mathcal{T}$-free $b$-factor polytope. 
Introducing ${\rm proj}_E(Q)$ and considering Proposition~\ref{prop:main}, which is a stronger statement, is a key idea in our proof. 
We also note that our algorithm in Section~\ref{sec:algo}
is based on the fact that the $\mathcal{T}$-free $b$-factor polytope is equal to ${\rm proj}_E(P)$. 
In this sense, both ${\rm proj}_E(P)$ and ${\rm proj}_E(Q)$ play important roles in this paper.

\begin{example}\label{exam:02}
Suppose that $G=(V, E)$, $b \in \mathbf{Z}^V_{\ge 0}$, and $x \in \mathbf{R}^E$ are as in Example~\ref{exam:01}. 
Let $T$ be the central triangle in $G$ and let $E(T) = \{\alpha, \beta, \gamma\}$. 
If $y \in \mathbf{R}^Y$ satisfies (\ref{eq:04}) and (\ref{eq:06}), then 
$y_{\alpha \beta} + y_{\beta \gamma} + y_{\alpha \gamma} \le  1$. 
Thus, without loss of generality, we may assume that $y_{\alpha \beta} \le \frac{1}{3}$ by symmetry. 
Let $S$ be a vertex set with $\delta(S) = \{\alpha, \beta\}$. 
Then, (\ref{eq:100}) does not hold for 
$(S, \{ \alpha \}, \{ \beta \}) \in \mathcal{F}$, 
because 
$x(\alpha)+(1-x(\beta)) - 2 q(T) = 1-x(\alpha)-x(\beta)+2 y_{\alpha \beta} \le \frac{2}{3} < 1$. 
Therefore, $x$ is not in ${\rm proj}_E(Q)$. 
\end{example}

\section{Extreme Points of the Projection of $Q$}
\label{sec:propertyQ}

In this section, we show a property of 
extreme points of ${\rm proj}_E (Q)$, which will be used in Section~\ref{sec:proof}. 
We begin with the following easy lemma. 

\begin{lemma}\label{lem:pre03}
Suppose that $x \in \mathbf{R}^E$ satisfies (\ref{eq:02}) and (\ref{eq:triangle}). 
Then, there exists $y \in \mathbf{R}^Y$ that satisfies (\ref{eq:04})--(\ref{eq:06}). 
\end{lemma}

\begin{proof}
Let $T \in \mathcal{T}$ be a triangle with $E(T) = \{\alpha, \beta, \gamma \}$ and $x(\alpha) \ge x(\beta) \ge x(\gamma)$. 
For $J \in \mathcal{E}_T$, we define $y(T, J)$ as follows. 
\begin{itemize}
\item
If $x(\alpha) \ge x(\beta) + x(\gamma)$, then 
$y_{\alpha \beta} = x(\beta)$, $y_{\alpha \gamma} = x(\gamma)$, $y_\emptyset = 1-x(\alpha)$, 
$y_{\alpha} = x(\alpha) - x(\beta) - x(\gamma)$, and $y_\beta = y_\gamma = y_{\beta \gamma} = 0$. 
\item
If  $x(\alpha) < x(\beta) + x(\gamma)$, then 
$y_{\alpha \beta} = \frac{1}{2} (x(\alpha) + x(\beta) - x(\gamma))$, 
$y_{\alpha \gamma} = \frac{1}{2} (x(\alpha) + x(\gamma) - x(\beta))$, 
$y_{\beta \gamma} = \frac{1}{2} (x(\beta) + x(\gamma) - x(\alpha))$, 
$y_\emptyset = 1- \frac{1}{2} (x(\alpha) + x(\beta) + x(\gamma))$, and 
$y_\alpha = y_\beta = y_\gamma = 0$. 
\end{itemize}
Then, $y$  satisfies (\ref{eq:04})--(\ref{eq:06}). 
\end{proof}

By using this lemma, we show the following. 

\begin{lemma}\label{lem:pre04}
Let $x$ be an extreme point of ${\rm proj}_E (Q)$
and $y \in \mathbf{R}^Y$ be a vector with $(x, y) \in Q$. 
Then, one of the following holds. 
\begin{enumerate}
\item[(i)]
$x=x_M$ for some $\mathcal{T}$-free $b$-factor $M \subseteq E$. 
\item[(ii)]
(\ref{eq:triangle}) is tight for some $T \in \mathcal{T}$. 
\item[(iii)]
(\ref{eq:100}) is tight for some $(S, F_0, F_1) \in \mathcal{F}$ with $\mathcal{T}^+_S \not= \emptyset$, 
where we define $\mathcal{T}^+_S = \{T \in \mathcal{T} \mid E(T) \cap \delta(S) \cap F_1 \not= \emptyset\}$. 
\end{enumerate}
\end{lemma}

\begin{proof}
We prove (i) by assuming that (ii) and (iii) do not hold. 
Since (\ref{eq:100}) is not tight for any $(S, F_0, F_1) \in \mathcal{F}$ with $\mathcal{T}^+_S \not= \emptyset$, 
$x$ is an extreme point of 
$$
\{ x \in \mathbf{R}^E \mid 
\mbox{There exists $y \in \mathbf{R}^Y$ such that $(x, y)$ satisfies (\ref{eq:01})--(\ref{eq:06})} \}, 
$$
because (\ref{eq:03}) is a special case of (\ref{eq:100}) in which $\mathcal{T}^+_S = \emptyset$. 
By Lemma~\ref{lem:pre03}, 
this polytope is equal to 
$\{ x \in \mathbf{R}^E \mid \mbox{$x$ satisfies (\ref{eq:01})--(\ref{eq:triangle})} \}$. 
Since (\ref{eq:triangle}) is not tight for any $T \in \mathcal{T}$, 
$x$ is an extreme point of 
$\{ x \in \mathbf{R}^E \mid \mbox{$x$ satisfies (\ref{eq:01})--(\ref{eq:03})} \}$, 
which is the $b$-factor polytope. 
Thus, $x$ is a characteristic vector of a $b$-factor. 
Since $x$ satisfies (\ref{eq:triangle}), 
it holds that $x=x_M$ for some $\mathcal{T}$-free $b$-factor $M \subseteq E$. 
\end{proof}


\section{Proof of Proposition~\ref{prop:main}}
\label{sec:proof}

In this section, we prove Proposition~\ref{prop:main}
by induction on $|\mathcal{T}|$. 
If $|\mathcal{T}|=0$, then
$y$ does not exist and
(\ref{eq:100}) is equivalent to (\ref{eq:03}). 
Thus, ${\rm proj}_E (Q)$ is the $b$-factor polytope, which shows the base case of the induction. 

Fix an instance $(G, b, \mathcal{T})$ with $|\mathcal{T}| \ge 1$ and 
assume that Proposition~\ref{prop:main} holds for instances with smaller $|\mathcal{T}|$. 
Suppose that $Q \not= \emptyset$, which implies that $b(V)$ is even as $(V, \emptyset, \emptyset) \not\in \mathcal{F}$ by (\ref{eq:100}).  
Pick up $x \in {\rm proj}_E (Q)$ and 
let $y \in \mathbf{R}^Y$ be a vector with $(x, y) \in Q$. 
Our aim is to show that $x$ is contained in the $\mathcal{T}$-free $b$-factor polytope. 

In what follows in this section, we prove Proposition~\ref{prop:main} as follows. 
We apply Lemma~\ref{lem:pre04} to obtain one of (i), (ii), and (iii). 
If (i) holds, that is, $x=x_M$ for some $\mathcal{T}$-free $b$-factor $M \subseteq E$, then
$x$ is obviously in the $\mathcal{T}$-free $b$-factor polytope. 
If (ii) holds, that is, (\ref{eq:triangle}) is tight for some $T \in \mathcal{T}$, then
we replace $T$ with a certain graph and apply the induction, which will be discussed in Section~\ref{sec:case2}. 
If (iii) holds, that is, (\ref{eq:100}) is tight for some $(S, F_0, F_1) \in \mathcal{F}$ with $\mathcal{T}^+_S \not= \emptyset$, then
we divide $G$ into two graphs and apply the induction for each graph, which will be discussed in Section~\ref{sec:case3}.


\subsection{When (\ref{eq:triangle}) is Tight} 
\label{sec:case2}

In this subsection, we consider the case when
(\ref{eq:triangle}) is tight for some $T \in \mathcal{T}$. 
Fix a triangle $T \in \mathcal{T}$ with $x(E(T)) = 2$, 
where we denote $V(T) = \{v_1, v_2, v_3\}$, $E(T) = \{\alpha, \beta, \gamma\}$, $\alpha = v_1 v_2$, $\beta = v_2 v_3$, and $\gamma = v_3 v_1$ (Figure~\ref{fig:101}). 
Since (\ref{eq:triangle}) is tight, we obtain
\begin{align*}
2 = x(\alpha) + x(\beta) +x(\gamma) = 2 (y_{\alpha \beta} + y_{\alpha \gamma} + y_{\beta \gamma}) + y_\alpha + y_\beta + y_\gamma 
   = 2 - (y_\alpha + y_\beta + y_\gamma) - 2 y_\emptyset, 
\end{align*}
and hence $y_\alpha = y_\beta = y_\gamma =  y_\emptyset = 0$. 
Therefore, 
$x(\alpha) = y_{\alpha \beta} + y_{\alpha \gamma}$, 
$x(\beta) = y_{\alpha \beta} + y_{\beta \gamma}$,  
$x(\gamma) = y_{\alpha \gamma} + y_{\beta \gamma}$, and 
$y_{\alpha \beta} + y_{\alpha \gamma} + y_{\beta \gamma} = 1$. 

We construct a new instance of the $\mathcal{T}$-free $b$-factor problem as follows. 
Let $G'=(V', E')$ be the graph obtained from $G = (V, E)$ by 
removing $E(T)$ and adding a new vertex $r$ together with three new edges $e_1=r v_1$, $e_2 = r v_2$, and $e_3 = r v_3$ as in Figure~\ref{fig:101}. 
Define $b' \in  \mathbf{Z}_{\ge 0}^{V'}$ as
$b'(r)=1$,  
$b'(v)=b(v)-1$ for $v \in \{v_1, v_2, v_3\}$, and 
$b'(v)=b(v)$ for $v \in V \setminus \{v_1, v_2, v_3\}$. 
Define $x' \in \mathcal \mathbf{R}^{E'}$ as 
$x'(e_1)=y_{\alpha \gamma}$, 
$x'(e_2)=y_{\alpha \beta}$, 
$x'(e_3)=y_{\beta \gamma}$, and
$x'(e)=x(e)$ for $e \in E' \cap E$.  
Let $\mathcal{T}' = \mathcal{T} \setminus \{T\}$, and
let $Y'$ and $\mathcal{F}'$ be the objects for the obtained instance $(G', b', \mathcal{T}')$
that are defined in the same way as $Y$ and $\mathcal{F}$. 
Define $y'$ as the restriction of $y$ to $Y'$. 
We now show the following claim.

\begin{figure}
\begin{center}
\includegraphics[clip,width=7cm]{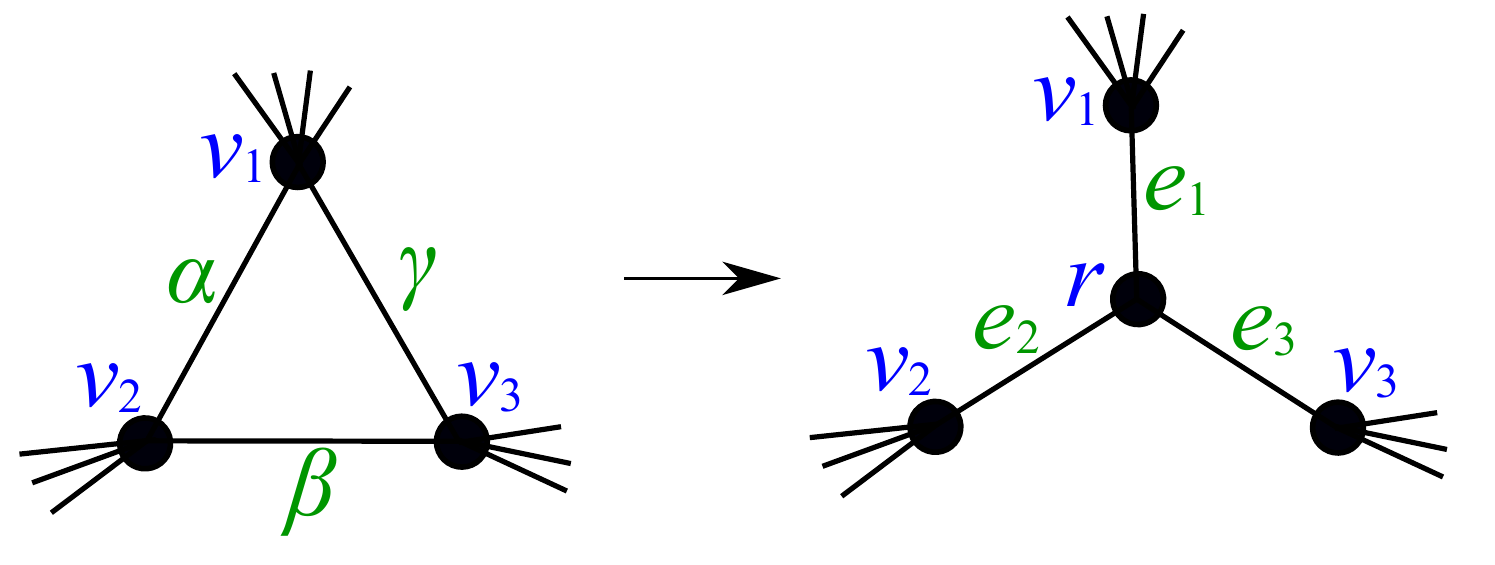}
\caption{Construction of $G'$}
\label{fig:101}
\end{center}
\end{figure}

\begin{claim}
\label{clm:30}
$(x', y')$ satisfies (\ref{eq:01}), (\ref{eq:02}), (\ref{eq:triangle})--(\ref{eq:06}), and (\ref{eq:100})
with respect to the new instance $(G', b', \mathcal{T}')$. 
\end{claim}

\begin{proof}
We can easily see that $(x', y')$ satisfies (\ref{eq:01}), (\ref{eq:02}), (\ref{eq:triangle})--(\ref{eq:06}). 
Consider (\ref{eq:100}) for 
$(S', F'_0, F'_1) \in \mathcal{F}'$. 
By changing the roles of $S'$ and $V' \setminus S'$ if necessary, 
we may assume that $r \in S'$. 
For $(S', F'_0, F'_1) \in \mathcal{F}'$ (resp.~$(S, F_0, F_1) \in \mathcal{F}$), 
we denote the left-hand side of (\ref{eq:100}) by $h'(S', F'_0, F'_1)$ (resp.~$h(S, F_0, F_1)$). 

Then, we obtain $h'(S', F'_0, F'_1) \ge 1$ for each $(S', F'_0, F'_1) \in \mathcal{F}'$
by the following case analysis and by the symmetry of $v_1$, $v_2$, and $v_3$. 
\begin{enumerate}
\item
Suppose that $v_1, v_2, v_3 \in S'$. 
Since $(S' \setminus \{r\}, F'_0, F'_1) \in \mathcal{F}$, 
we obtain 
$h'(S', F'_0, F'_1) = h(S' \setminus \{r\}, F'_0, F'_1) \ge 1$. 

\item
Suppose that $v_2, v_3 \in S'$ and $v_1 \not\in S'$. 
\begin{itemize}
\item
If $e_1 \in F'_0$, then define $(S, F_0, F_1) \in \mathcal{F}$ as 
$S =S' \setminus \{r\}$, $F_0 = (F'_0 \setminus \{e_1\}) \cup \{\alpha\}$, and $F_1 = F'_1 \cup\{\gamma\}$. 
Since $x(\alpha) + (1- x(\gamma)) - 2 q(T) = y_{\alpha \gamma} = x'(e_1)$, 
we obtain 
$h'(S', F'_0, F'_1) = h(S, F_0, F_1) \ge 1$. 
\item
If $e_1 \in F'_1$, then define $(S, F_0, F_1) \in \mathcal{F}$ as 
$S =S' \setminus \{r\}$, $F_0 = F'_0$, and $F_1 = (F'_1 \setminus \{e_1\}) \cup\{\alpha, \gamma\}$. 
Since $(1-x(\alpha)) + (1- x(\gamma)) - 2 q(T) = y_{\beta \gamma} + y_{\alpha \beta} = 1 - x'(e_1)$, 
we obtain 
$h'(S', F'_0, F'_1) = h(S, F_0, F_1) \ge 1$. 
\end{itemize}

\item
Suppose that $v_1 \in S'$ and $v_2, v_3 \not\in S'$. 
\begin{itemize}
\item
If $e_2, e_3 \in F'_0$, then define $(S, F_0, F_1) \in \mathcal{F}$ as 
$S =S' \setminus \{r\}$, $F_0 = F'_0 \setminus \{e_2, e_3\}$, and $F_1 = F'_1 \cup\{\alpha, \gamma\}$. 
Since $(1-x(\alpha)) + (1- x(\gamma)) - 2 q(T) = y_{\beta \gamma} + y_{\alpha \beta} = x'(e_2) + x'(e_3)$, 
we obtain 
$h'(S', F'_0, F'_1) = h(S, F_0, F_1) \ge 1$. 

\item
If $e_2 \in F'_0$ and $e_3 \in F'_1$, then define $(S, F_0, F_1) \in \mathcal{F}$ as
$S =S' \setminus \{r\}$, $F_0 = (F'_0 \setminus \{e_2\}) \cup \{\alpha\}$, and $F_1 = (F'_1 \setminus \{e_3\}) \cup\{\gamma\}$. 
Since $x(\alpha) + (1- x(\gamma)) - 2 q(T) = y_{\alpha \gamma} \le 1- x'(e_3)$, 
we obtain 
$h'(S', F'_0, F'_1) \ge h(S, F_0, F_1) \ge 1$. 

\item
If $e_2, e_3 \in F'_1$, then 
$h'(S', F'_0, F'_1) \ge 2 - x'(e_2) - x'(e_3) \ge 1$. 
\end{itemize}

\item
Suppose that $v_1, v_2, v_3 \not\in S'$. 
\begin{itemize}
\item
If $F'_1 \cap \delta_{G'}(r) = \emptyset$, then 
$h'(S', F'_0, F'_1) \ge x'(e_1) + x'(e_2) + x'(e_3) =1$.
\item
If $|F'_1 \cap \delta_{G'}(r)| \ge 2$, then 
$h'(S', F'_0, F'_1) \ge 2 - (x'(e_1) + x'(e_2) + x'(e_3)) =1$.
\item
If $|F'_1 \cap \delta_{G'}(r)| = 1$, then 
define $(S, F_0, F_1) \in \mathcal{F}$ as 
$S =S' \setminus \{r\}$, $F_0 = F'_0 \setminus \delta_{G'}(r)$, and $F_1 = F'_1 \setminus \delta_{G'}(r)$. 
Then, we obtain 
$h'(S', F'_0, F'_1) \ge h(S, F_0, F_1) \ge 1$. 
\end{itemize}
\end{enumerate}
\end{proof}

By this claim and by the induction hypothesis, 
$x'$ is in the $\mathcal{T}'$-free $b'$-factor polytope. 
That is, there exist $\mathcal{T}'$-free $b'$-factors $M'_1, \dots , M'_t$ in $G'$ and 
non-negative coefficients $\lambda_1, \dots , \lambda_t$ such that $\sum_{i=1}^t \lambda_i = 1$ and 
\begin{equation}
x' = \sum_{i=1}^t \lambda_i x_{M'_i}, \label{eq:20}
\end{equation}
where $x_{M'_i} \in \mathbf{R}^{E'}$ is the characteristic vector of $M'_i$ defined in the same way as (\ref{eq:chara}). 

For a $\mathcal{T}'$-free $b'$-factor $M' \subseteq E'$ in $G'$, 
we define a corresponding $\mathcal{T}$-free $b$-factor $\varphi(M') \subseteq E$ in $G$ 
as
$$
\varphi(M') = 
\begin{cases}
(M' \cap E) \cup \{\alpha, \gamma\} & \mbox{if $e_1 \in M'$}, \\
(M' \cap E) \cup \{\alpha, \beta\} & \mbox{if $e_2 \in M'$}, \\
(M' \cap E) \cup \{\beta, \gamma\} & \mbox{if $e_3 \in M'$}. 
\end{cases}
$$
By (\ref{eq:20}), we obtain 
$x(e) = \sum_{i=1}^t \lambda_i x_{\varphi(M'_i)}(e)$
for each $e \in E \cap E'$. 
By (\ref{eq:20}) again, it holds that 
$$
\sum_{i=1}^t \lambda_i x_{\varphi(M'_i)}(\alpha) 
=  \sum \{ \lambda_i  \mid e_1 \in M'_i  \} + \sum \{ \lambda_i  \mid e_2 \in M'_i  \}
= x'(e_1) + x'(e_2) = y_{\alpha \gamma} + y_{\alpha \beta}  = x(\alpha), 
$$
and similar equalities hold for $\beta$ and $\gamma$. 
Therefore, we obtain 
$x = \sum_{i=1}^t \lambda_i x_{\varphi(M'_i)}$, 
which shows that $x$ is in the $\mathcal{T}$-free $b$-factor polytope.


\subsection{When (\ref{eq:100}) is Tight}
\label{sec:case3}

In this subsection, we consider the case when (\ref{eq:100}) is tight 
for $(S^*, F^*_0, F^*_1) \in \mathcal{F}$ with $\mathcal{T}^+_{S^*} \not= \emptyset$, 
where $\mathcal{T}^+_{S^*} = \{T \in \mathcal{T} \mid E(T) \cap \delta(S^*) \cap F^*_1 \not= \emptyset\}$. 
In this case, we divide the original instance into two instances $(G_1, b_1, \mathcal{T}_1)$ and $(G_2, b_2, \mathcal{T}_2)$, apply the induction for each instance, 
and combine the two parts. 
We denote  
$\tilde{F}^*_0 = F^*_0 \setminus \bigcup_{T \in \mathcal{T}^+_{S^*}} E(T)$ and
$\tilde{F}^*_1 = F^*_1 \setminus \bigcup_{T \in \mathcal{T}^+_{S^*}} E(T)$. 


\subsubsection{Construction of $(G_j, b_j, \mathcal{T}_j)$}

We first construct $(G_1, b_1, \mathcal{T}_1)$ and its feasible LP solution $x_1$. 
Starting from the subgraph $G[S^*] = (S^*, E[S^*])$ 
induced by $S^*$,   
we add a new vertex $r$ corresponding to $V^* \setminus S^*$, set $b_1(r) =1$,  
and apply the following procedure. 
\begin{itemize}
\item
For each $f =uv \in \tilde{F}^*_0$ with $u \in S^*$, we add a new edge $e^f = u r$ (Figure~\ref{fig:103}). 
Let $x_1(e^f) = x(f)$.  
\item
For each $f =uv \in \tilde{F}^*_1$ with $u \in S^*$, 
we add a new vertex $p^f_u$ and new edges $e^f_u=u p^f_u$ and $e^f_r= p^f_u r$ (Figure~\ref{fig:104}). 
Let $b_1(p^f_u) = 1$, $x_1(e^f_u) = x(f)$, and $x_1(e^f_r) = 1 - x(f)$.  
\item
For each $T \in \mathcal{T}^+_{S^*}$ with $|E(T) \cap \delta_G(S^*) \cap F^*_1| = 2$ and $|V(T) \cap S^*| =2$, 
which we call {\it a triangle of type (A)},  
add new vertices $p_1, p_2$ and new edges $e_1, \dots , e_6$ as in Figure~\ref{fig:105}. 
Define $b_1(p_1) = b_1(p_2)=1$ and 
\begin{align*}
&x_1(e_1) = y_\emptyset +  y_\gamma, &  &x_1(e_2) = y_\emptyset +  y_\alpha,  &  &x_1(e_3) = y_{\alpha \beta},\\ 
&x_1(e_4) = y_{\beta \gamma}, &  &x_1(e_5) = 1-y_\emptyset - y_\gamma, &  & x_1(e_6) = 1-y_\emptyset - y_\alpha,
\end{align*}
where $\alpha, \beta$, and $\gamma$ are as in Figure~\ref{fig:105}.

\item
For each $T \in \mathcal{T}^+_{S^*}$ with $|E(T) \cap \delta_G(S^*) \cap F^*_1| = 2$ and $|V(T) \cap S^*| =1$,  
which we call {\it a triangle of type (A')},  
add a new vertex $p_3$ and new edges $e_1, e_2, e_3, e_4, e_7, e_8$, and $e_9$ as in Figure~\ref{fig:106}. 
Define $b_1(p_3) = 2$ and 
\begin{align*}
&x_1(e_1) = y_\emptyset +  y_\gamma, &  &x_1(e_2) = y_\emptyset +  y_\alpha,  &  &x_1(e_3) = y_{\alpha \beta}, & &x_1(e_4) = y_{\beta \gamma},\\ 
&x_1(e_7) = y_\beta, &  &x_1(e_8) = y_{\alpha \gamma}, &  & x_1(e_9) = 1-y_\emptyset - y_\beta,
\end{align*}
where $\alpha, \beta$, and $\gamma$ are as in Figure~\ref{fig:106}.

\item
For each $T \in \mathcal{T}^+_{S^*}$ with $|E(T) \cap \delta_G(S^*) \cap F^*_1| = 1$ and $|V(T) \cap S^*| =2$,  
which we call {\it a triangle of type (B)},  
add new vertices $p_1, p_2 , p_3$ and new edges $e_1, \dots , e_{9}$ as in Figure~\ref{fig:107}. 
Define $b_1(p_i) = 1$ for $i \in \{1, 2, 3\}$,  and 
\begin{align*}
&x_1(e_1) = y_\emptyset +  y_{\beta}, &  &x_1(e_2) = y_{\alpha \gamma},  &  &x_1(e_3) = y_\gamma,\\ 
&x_1(e_4) = y_{\alpha}, &  &x_1(e_5) = y_{\alpha \beta}, &  &x_1(e_6) = y_{\beta \gamma}, \\ 
&x_1(e_7) = y_{\emptyset}, &  &x_1(e_8) = 1- y_\emptyset - y_\gamma, &  &x_1(e_9) = 1-y_\emptyset - y_{\alpha}, 
\end{align*}
where $\alpha, \beta$, and $\gamma$ are as in Figure~\ref{fig:107}. 

\item
For each $T \in \mathcal{T}^+_{S^*}$ with $|E(T) \cap \delta_G(S^*) \cap F^*_1| = 1$ and $|V(T) \cap S^*| =1$,  
which we call {\it a triangle of type (B')},  
add a new vertex $p_4$ and new edges $e_1, e_2$, and $e_{10}$ as in Figure~\ref{fig:108}. 
Define $b_1(p_4) = 1$, and 
\begin{align*}
&x_1(e_1) = y_\emptyset +  y_{\beta}, &  &x_1(e_2) = y_{\alpha \gamma},  &  &x_1(e_{10}) = 1 - y_\emptyset - y_\beta, 
\end{align*}
where $\alpha, \beta$, and $\gamma$ are as in Figure~\ref{fig:108}. 
\end{itemize}
In order to make it clear that $p_i$ and $e_i$ are associated with $T\in \mathcal{T}^+_{S^*}$, 
we sometimes denote $p^{T}_i$ and $e^{T}_i$. 
Let $G_1=(V_1, E_1)$ be the obtained graph.
Define $b_1 \in \mathbf{Z}_{\ge 0}^{V_1}$ by 
$b_1(v)=b(v)$ for $v \in S^*$ and 
$b_1(v)$ is as above for $v \in V_1 \setminus S^*$. 
Define $x_1 \in \mathbf{R}^{E_1}$ by
$x_1(e)=x(e)$ for $e \in E[S^*]$ and 
$x_1(e)$ is as above for $e \in E_1 \setminus E[S^*]$.

For each $T \in \mathcal{T}_{S^*} \setminus \mathcal{T}^+_{S^*}$ with $|V(T) \cap S^*| =2$, say $V(T) \cap S^* = \{u, v\}$, 
let $\psi(T)$ be the corresponding triangle in $G_1$ whose vertex set is $\{u, v, r\}$. 
Let 
$$
\mathcal{T}_1 = \{T \in \mathcal T \mid V(T) \subseteq S^* \} \cup \{\psi(T) \mid T \in \mathcal{T}_{S^*} \setminus \mathcal{T}^+_{S^*} \mbox{ with } |V(T) \cap S^*| =2 \}, 
$$
and 
let $Y_1$ and $\mathcal{F}_1$ be the objects for the obtained instance $(G_1, b_1, \mathcal{T}_1)$
that are defined in the same way as $Y$ and $\mathcal{F}$. 
Define $y_1$ as the restriction of $y$ to $Y_1$, 
where we identify $f \in F^*_0$ with $e^f$ and 
identify $T \in \mathcal{T}_{S^*} \setminus \mathcal{T}^+_{S^*}$ with $\psi(T)$.

\begin{figure}
 \begin{minipage}{0.48\hsize}
\begin{center}
\includegraphics[clip,width=4cm]{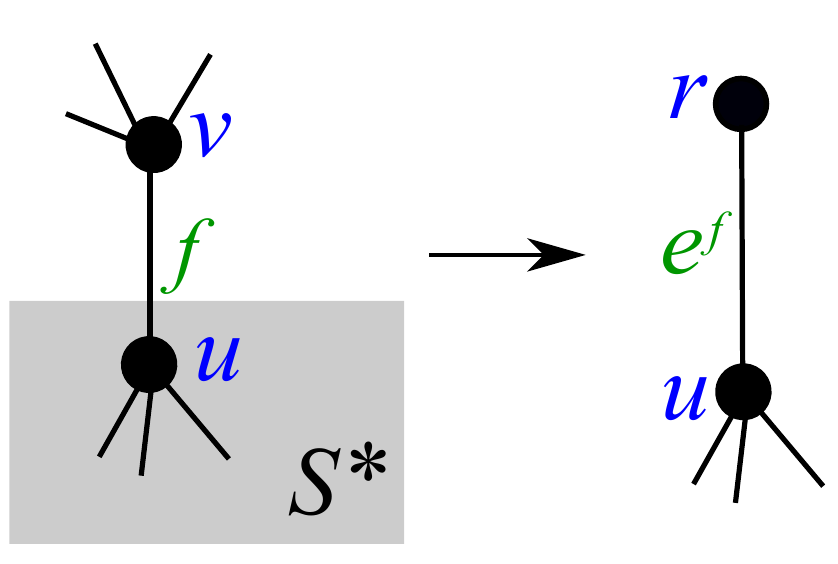}
\caption{An edge in $\tilde{F}^*_0$}
\label{fig:103}
\end{center}
 \end{minipage}
 \begin{minipage}{0.48\hsize}
\begin{center}
\includegraphics[clip,width=4cm]{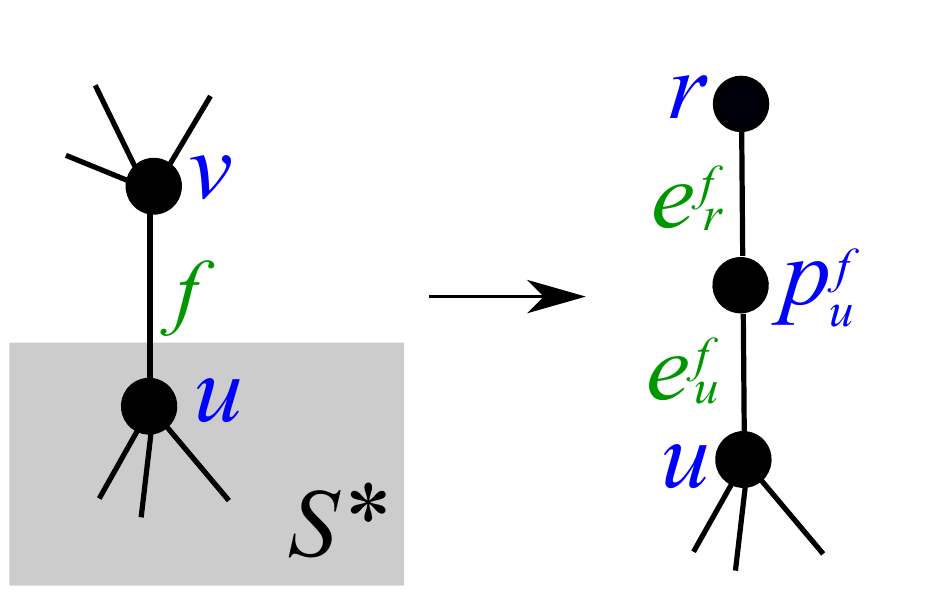}
\caption{An edge in $\tilde{F}^*_1$}
\label{fig:104}
\end{center}
 \end{minipage}
\end{figure}

\begin{figure}
 \begin{minipage}{0.60\hsize}
\begin{center}
\includegraphics[clip,width=7cm]{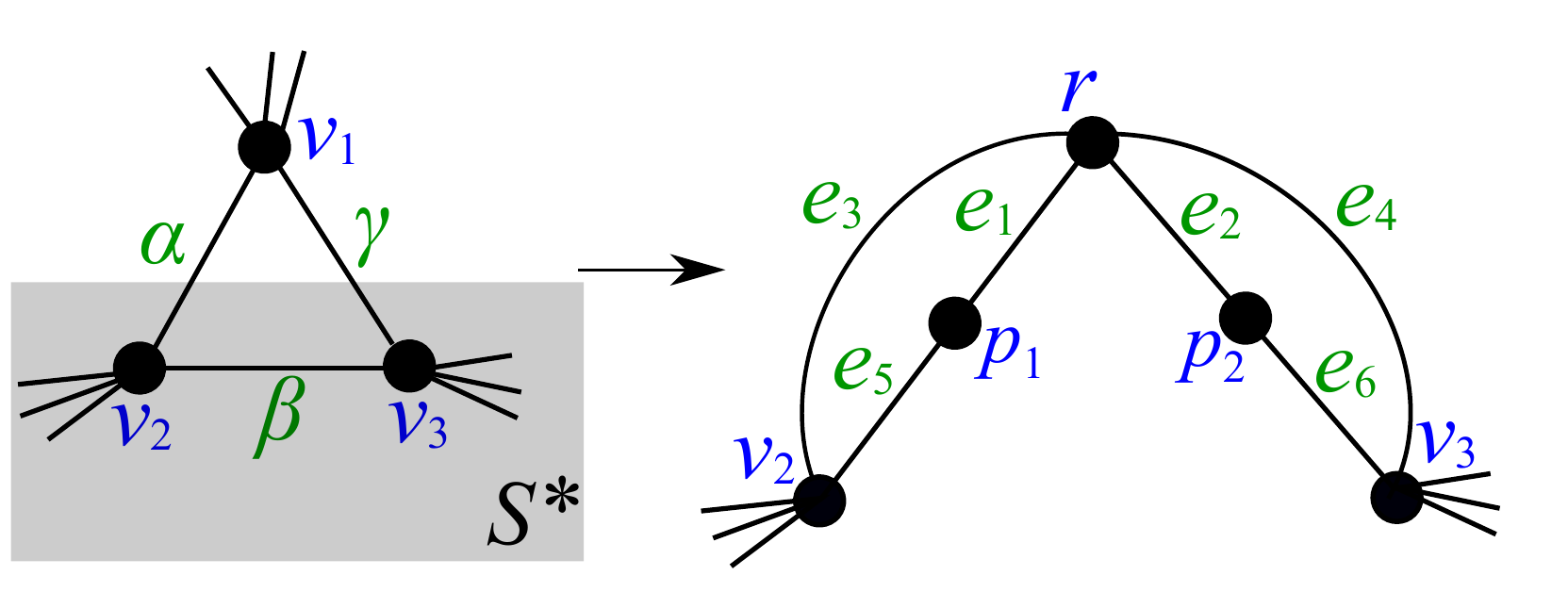}
\caption{A triangle of type (A)}
\label{fig:105}
\end{center}
 \end{minipage}
 \begin{minipage}{0.38\hsize}
\begin{center}
\includegraphics[clip,width=5.3cm]{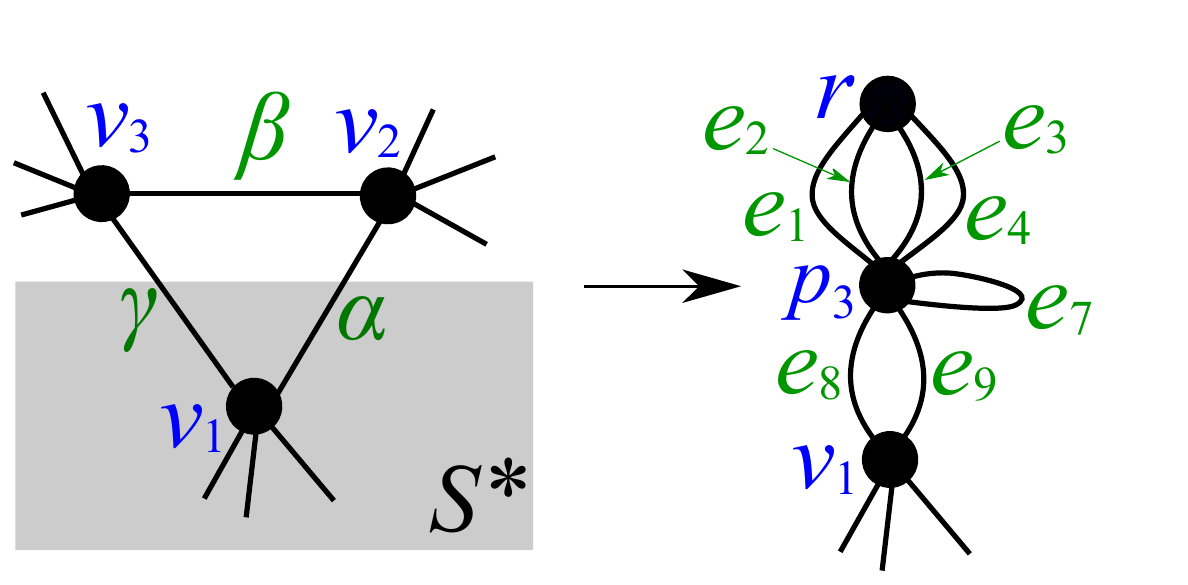}
\caption{A triangle of type (A')}
\label{fig:106}
\end{center}
 \end{minipage}
\end{figure}

\begin{figure}
 \begin{minipage}{0.60\hsize}
\begin{center}
\includegraphics[clip,width=7cm]{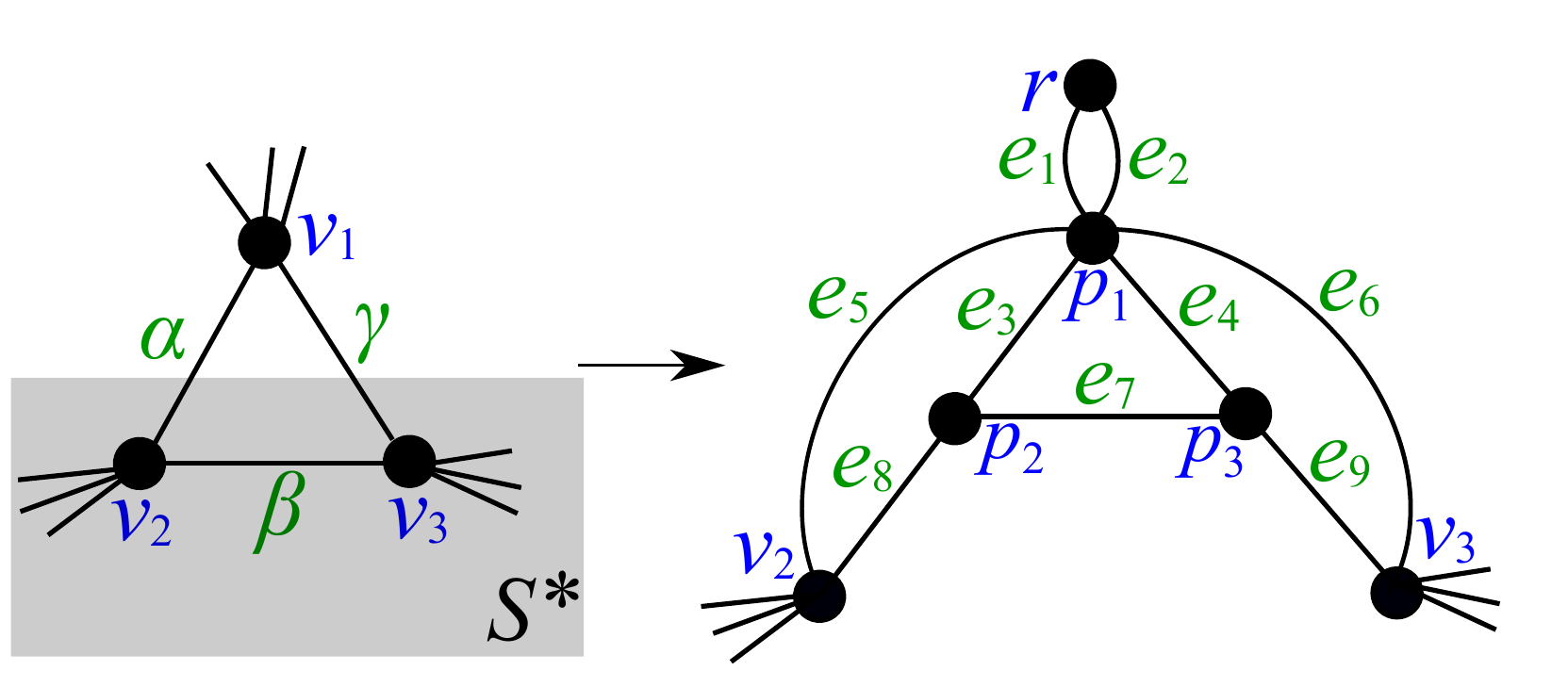}
\caption{A triangle of type (B)}
\label{fig:107}
\end{center}
 \end{minipage}
 \begin{minipage}{0.38\hsize}
\begin{center}
\includegraphics[clip,width=5cm]{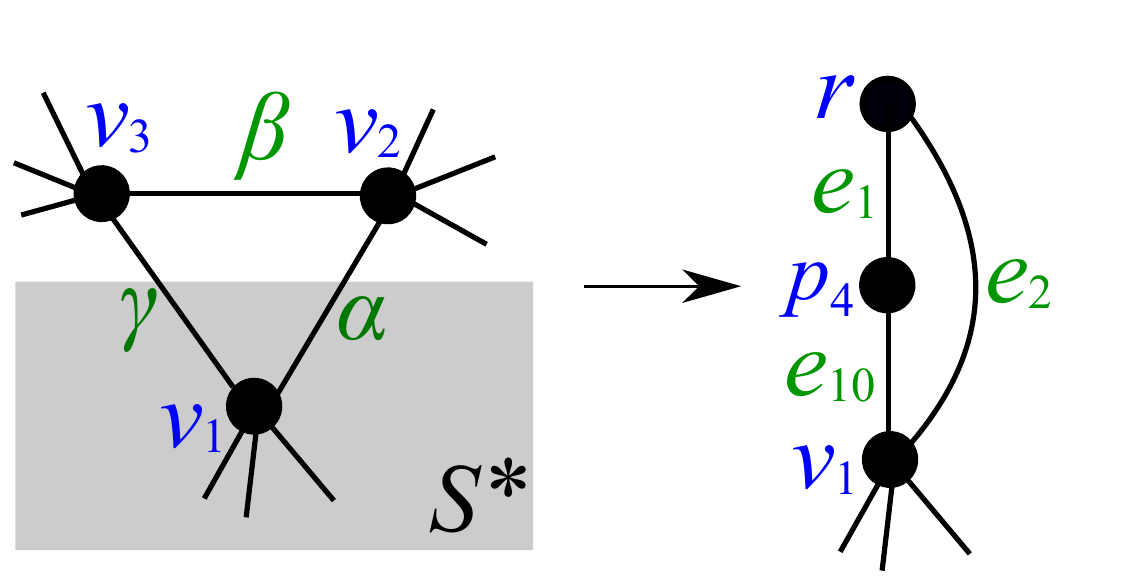}
\caption{A triangle of type (B')}
\label{fig:108}
\end{center}
 \end{minipage}
\end{figure}

Similarly, by changing the roles of $S^*$ and $V \setminus S^*$, 
we construct a graph $G_2=(V_2, E_2)$  and an instance $(G_2, b_2, \mathcal{T}_2)$, 
where the new vertex corresponding to  $S^*$ is denoted by $r'$.  
Define $x_2, y_2, Y_2$, and $\mathcal{F}_2$ in the same way as above. 
Note that a triangle $T \in \mathcal{T}^+_{S^*}$ is of type (A) (resp.~type (B)) for $(G_1, b_1, \mathcal{T}_1)$ 
if and only if it is of type (A') (resp.~type (B')) for $(G_2, b_2, \mathcal{T}_2)$. 

We use the following claim, whose proof is given in Appendix~\ref{sec:prf50}. 

\begin{claim}
\label{clm:key01}
For $j \in \{1, 2\}$, $(x_j, y_j)$ satisfies (\ref{eq:01}), (\ref{eq:02}), (\ref{eq:triangle})--(\ref{eq:06}), and (\ref{eq:100})
with respect to the new instance $(G_j, b_j, \mathcal{T}_j)$. 
\end{claim}


\subsubsection{Pairing up $\mathcal{T}_j$-free $b_j$-factors}

Since $|\mathcal{T}_j| \le |\mathcal{T}| - |\mathcal{T}^+_{S^*}| < |\mathcal{T}|$ for $j \in \{1, 2\}$, 
by Claim~\ref{clm:key01} and by the induction hypothesis, 
$x_j$ is in the $\mathcal{T}_j$-free $b_j$-factor polytope. 
That is, there exists a set $\mathcal{M}_j$ of $\mathcal{T}_j$-free $b_j$-factors in $G_j$ and 
a non-negative coefficient $\lambda_{M}$ for each $M \in \mathcal{M}_j$ such that $\sum_{M \in \mathcal{M}_j} \lambda_M = 1$ and
$x_j = \sum_{M \in \mathcal{M}_j} \lambda_M x_{M}$, 
where $x_{M} \in \mathbf{R}^{E_j}$ is the characteristic vector of $M$.  

Let $j \in \{1, 2\}$ and consider $(G_j, b_j, \mathcal{T}_j)$. 
Since $x_j(e^T_1) \ge x_j(e^T_7)$ for each triangle $T \in \mathcal{T}^+_{S^*}$ of type (B), 
by swapping parallel edges $e^T_1$ and $e^T_2$ if necessary, 
we may assume that  
$\{e^T_2, e^T_7\} \not \subseteq M$ for each $M \in \mathcal{M}_j$ and for each $T \in \mathcal{T}^+_{S^*}$ of type (B).
In what follows, we construct a collection of $\mathcal{T}$-free $b$-factors in $G$ by combining $\mathcal{M}_1$ and $\mathcal{M}_2$. 

Since there is a one-to-one correspondence between $\delta_{G_1}(r)$ and $\delta_{G_2}(r')$, 
we identify them and denote $E_0$, that is, $E_0 = E_1 \cap E_2 =  \delta_{G_1}(r) = \delta_{G_2}(r')$. 
Note that $e^f_r \in E_1$ and $e^f_{r'} \in E_2$ are identified for each $f \in \tilde{F}^*_1$. 
Since $b_1(r) = b_2(r') = 1$, it holds that $|M_1 \cap E_0| = |M_2 \cap E_0| = 1$ for every $M_1 \in \mathcal{M}_1$ and for every $M_2 \in \mathcal{M}_2$. 
Define 
$$
\mathcal{M} = \{ (M_1, M_2) \mid M_1 \in \mathcal{M}_1, \ M_2 \in \mathcal{M}_2, \ M_1 \cap E_0 = M_2 \cap E_0\}. 
$$ 
Since $x_1(e) = x_2(e)$ for $e \in E_0$ by the definitions of $x_1$ and $x_2$, 
we can pair up a $b_1$-factor $M_1$ in $\mathcal{M}_1$ and a $b_2$-factor $M_2$ in $\mathcal{M}_2$ so that $(M_1, M_2) \in \mathcal{M}$. 
More precisely, 
we can assign a non-negative coefficient $\lambda_{(M_1, M_2)}$ for each pair $(M_1, M_2) \in \mathcal{M}$
such that 
\begin{align}
& \sum \{ \lambda_{(M_1, M_2)} \mid   (M_1, M_2) \in \mathcal{M} \} = 1, \label{eq:51} \\
& \sum \{ \lambda_{(M_1, M_2)} \mid   (M_1, M_2) \in \mathcal{M}, \ e' \in M_1 \} = x_1(e')  \qquad  (e' \in E_1), \label{eq:52} \\
& \sum \{ \lambda_{(M_1, M_2)} \mid   (M_1, M_2) \in \mathcal{M}, \ e' \in M_2 \} = x_2(e')  \qquad  (e' \in E_2). \label{eq:53}
\end{align}

Let $(M_1, M_2) \in \mathcal{M}$.  
For a triangle $T \in \mathcal{T}^+_{S^*}$ of type (A) or (A'), 
denote $M_T = (M_1 \cup M_2) \cap \{ e^T_1, \dots , e^T_{9}\}$ and define 
$\varphi(M_1, M_2, T) \subseteq E(T)$ as  
\begin{align*}
&\varphi(M_1, M_2, T) = \begin{cases}
\{\alpha\} & \mbox{if $M_T=\{e_2, e_5, e_8\}$ or $M_T=\{e_2, e_5, e_9\}$}, \\ 
\{\gamma \} & \mbox{if $M_T=\{e_1, e_6, e_8\}$ or $M_T=\{e_1, e_6, e_9\}$},  \\
\{\alpha, \beta\} & \mbox{if $M_T=\{e_3, e_5, e_6, e_8\}$ or $M_T=\{e_3, e_5, e_6, e_9\}$}, \\
\{\beta, \gamma\} & \mbox{if $M_T=\{e_4, e_5, e_6, e_8\}$ or $M_T=\{e_4, e_5, e_6, e_9\}$}, \\
\{\alpha, \gamma\} & \mbox{if $M_T = \{e_5, e_6, e_8, e_9\}$}, \\
\{\beta\} & \mbox{if $M_T = \{e_5, e_6, e_7\}$}, 
\end{cases}
\end{align*}
where the superscript $T$ is omitted here. 
Note that $M_T$ satisfies one of the above conditions,
because $M_j$ is a $b_j$-factor for $j \in \{1, 2\}$. 

For a triangle $T \in \mathcal{T}^+_{S^*}$ of type (B) or (B'), 
denote $M_T = (M_1 \cup M_2) \cap \{ e^T_1, \dots , e^T_{10}\}$ and define 
$\varphi(M_1, M_2, T) \subseteq E(T)$ as  
\begin{align*}
&\varphi(M_1, M_2, T) = \begin{cases}
\emptyset & \mbox{if $M_T = \{e_1, e_7\}$}, \\
\{\alpha\} & \mbox{if $M_T = \{e_4, e_8, e_{10} \}$ or $M_T = \{e_5, e_7, e_{10} \}$}, \\ 
\{\gamma \} & \mbox{if $M_T = \{e_3, e_9, e_{10} \}$ or $M_T = \{e_6, e_7, e_{10} \}$},  \\
\{\alpha, \beta\} & \mbox{if $M_T = \{e_5, e_8, e_9, e_{10} \}$}, \\
\{\beta, \gamma\} & \mbox{if $M_T = \{e_6, e_8, e_9, e_{10} \}$}, \\
\{\alpha, \gamma\} & \mbox{if $M_T = \{e_2, e_8, e_9, e_{10}\}$}, \\
\{\beta\} & \mbox{if $M_T = \{e_1, e_8, e_9\}$}, 
\end{cases}
\end{align*}
where the superscript $T$ is omitted here again. 
Note that $M_T$ satisfies one of the above conditions,
because we are assuming that $M_j$ is a $b_j$-factor 
with $\{e_2, e_7\} \not \subseteq M_j$ for $j \in \{1, 2\}$.

For $(M_1, M_2) \in \mathcal{M}$, 
define $M_1 \oplus M_2 \subseteq E$ as 
\begin{align*}
M_1 \oplus M_2 &= (M_1 \cap E[S^*]) \cup (M_2 \cap E[V \setminus S^*]) \cup \{ f \in \tilde{F}^*_0 \mid e^f \in M_1 \cap M_2 \} \\
                       & \quad \cup \{ f \in \tilde{F}^*_1 \mid e^f_r \not\in M_1 \cap M_2 \} 
                           \cup \bigcup \{ \varphi(M_1, M_2, T) \mid T \in \mathcal{T}^+_{S^*} \}. 
\end{align*}

We now use the following claims, whose proofs are postponed to Appendices~\ref{sec:key03} and~\ref{sec:key04}. 

\begin{claim}
\label{clm:key03}
For $(M_1, M_2) \in \mathcal{M}$, $M_1 \oplus M_2$ forms a $\mathcal{T}$-free $b$-factor. 
\end{claim}

\begin{claim}
\label{clm:key04}
It holds that
\begin{equation*}
x = \sum_{(M_1, M_2) \in \mathcal{M}}  \lambda_{(M_1, M_2)} x_{M_1 \oplus M_2}, 
\end{equation*}
where $x_{M_1 \oplus M_2} \in \mathbf{R}^{E}$ is the characteristic vector of $M_1 \oplus M_2$. 
\end{claim}

By (\ref{eq:51}) and by Claims~\ref{clm:key03} and~\ref{clm:key04}, 
it holds that $x$ is in the $\mathcal{T}$-free $b$-factor polytope. 
This completes the proof of Proposition~\ref{prop:main}.


\section{Algorithm}
\label{sec:algo}

In this section, we give a polynomial-time algorithm for the weighted $\mathcal{T}$-free $b$-factor problem and prove Theorem~\ref{thm:algo}. 
Our algorithm is based on the ellipsoid method using 
the fact that the $\mathcal{T}$-free $b$-factor polytope is equal to ${\rm proj}_E (P)$ 
(Theorem~\ref{thm:main01}). 
In order to apply the ellipsoid method, 
we need a polynomial-time algorithm for the separation problem. 
That is, for $(x, y) \in \mathbf{R}^E \times \mathbf{R}^Y$, 
we need a polynomial-time algorithm that concludes $(x, y) \in P$ or 
returns a violated inequality. 

Let $(x, y) \in \mathbf{R}^E \times \mathbf{R}^Y$. 
We can easily check whether $(x, y)$ satisfies  
(\ref{eq:01}), (\ref{eq:02}), and (\ref{eq:triangle})--(\ref{eq:06}) 
or not in polynomial time. 
In order to solve the separation problem for (\ref{eq:07}), 
we use the following theorem, 
which implies that 
the separation problem for (\ref{eq:03}) can be solved in polynomial time.

\begin{theorem}[Padberg-Rao~\cite{PR82} (see also~\cite{LRT08})]
\label{lem:oddmincut}
Suppose we are given a graph $G'=(V', E')$, $b' \in \mathbf{Z}_{\ge 0}^{V'}$, and $x' \in [0, 1]^{E'}$. 
Then, in polynomial time, we can compute $S' \subseteq V'$ and a partition $(F'_0, F'_1)$ of $\delta_{G'}(S')$ 
that minimize $\sum_{e \in F'_0} x'(e) + \sum_{e \in F'_1} (1-x'(e))$ subject to $b'(S') + |F'_1|$ is odd. 
\end{theorem}

In what follows, we reduce the separation problem for (\ref{eq:07}) to 
that for (\ref{eq:03}) and utilize Theorem~\ref{lem:oddmincut}. 
Suppose that $(x, y) \in \mathbf{R}^E \times \mathbf{R}^Y$ satisfies  
(\ref{eq:01}), (\ref{eq:02}), and (\ref{eq:triangle})--(\ref{eq:06}). 
For each triangle $T \in \mathcal{T}$,  
we remove $E(T)$ and add a vertex $r_T$ together with three new edges 
$e_1 = r_T v_1$, $e_2 = r_T v_2$, and $e_3 = r_T v_3$ (Figure~\ref{fig:102}). 
Let $E'_T = \{e_1, e_2, e_3\}$ and
define 
\begin{align}
&x'(e_1) = x(\alpha) + x(\gamma) - 2 y_{\alpha \gamma}, &
&x'(e_2) = x(\alpha) + x(\beta) - 2 y_{\alpha \beta},  &
&x'(e_3) = x(\beta) + x(\gamma) - 2 y_{\beta \gamma}. \label{eq:x'}
\end{align}
Let $G'=(V', E')$ be the graph obtained from $G$ by applying this procedure for every $T \in \mathcal{T}$. 
Define $b' \in \mathbf{Z}_{\ge 0}^{V'}$ as $b'(v)=b(v)$ for $v \in V$ and $b'(v)=0$ for $v \in V' \setminus V$. 
By setting $x'(e)=x(e)$ for $e \in E' \cap E$ 
and by defining $x'(e)$ as (\ref{eq:x'}) for $e \in E' \setminus E$, 
we obtain $x' \in [0, 1]^{E'}$. 
We now show the following lemma.

\begin{figure}
\begin{center}
\includegraphics[clip,width=7cm]{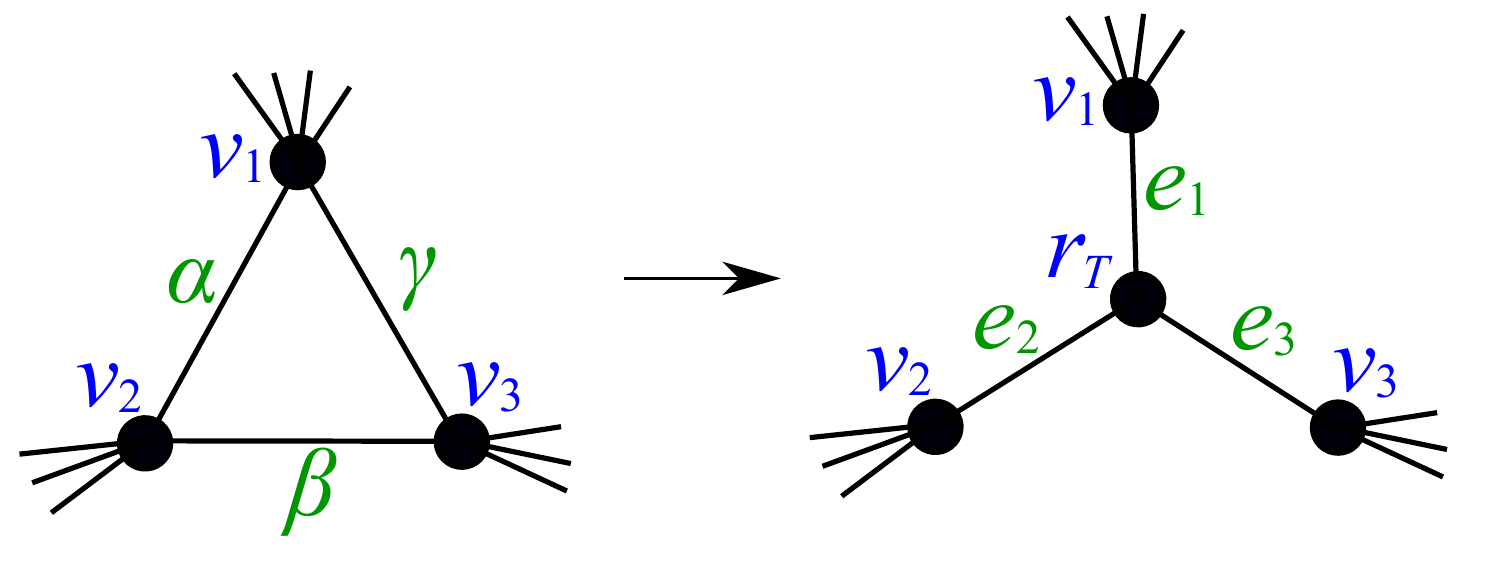}
\caption{Replacement of a triangle $T \in \mathcal{T}$}
\label{fig:102}
\end{center}
\end{figure}

\begin{lemma}
\label{lem:cutequiv}
Suppose that $(x, y) \in \mathbf{R}^E \times \mathbf{R}^Y$ satisfies  
(\ref{eq:01}), (\ref{eq:02}), and (\ref{eq:triangle})--(\ref{eq:06}). 
Define $G'=(V', E')$, $b'$, and $x'$ as above. 
Then, $(x, y)$ violates (\ref{eq:07}) for some $(S, F_0, F_1) \in \mathcal{F}$  
if and only if
there exist $S' \subseteq V'$ and a partition $(F'_0, F'_1)$ of $\delta_{G'}(S')$ 
such that $b'(S') + |F'_1|$ is odd and $\sum_{e \in F'_0} x'(e) + \sum_{e \in F'_1} (1-x'(e)) < 1$. 
\end{lemma}

\begin{proof}
First, to show the ``only if'' part, 
assume that $(x, y)$ violates (\ref{eq:07}) for some $(S, F_0, F_1) \in \mathcal{F}$. 
Recall that $\mathcal{T}_S = \{ T \in \mathcal{T} \mid E(T) \cap \delta_G(S) \not= \emptyset\}$. 
Define $S' \subseteq V'$ by 
$$
S' = S \cup \{r_T \mid T \in \mathcal{T},\ |V(T) \cap S| \ge 2\}. 
$$
Then, for each $T \in \mathcal{T}_S$, $E'_T \cap \delta_{G'}(S')$ consists of a single edge, which we denote $e_T$. 
Define $F'_0$ and $F'_1$ as follows: 
\begin{align*}
F'_0 &= (F_0 \cap E') \cup \{e_T \mid T \in \mathcal{T}_S, \ \mbox{$|E(T) \cap F_1| = 0$ or $2$} \}, \\ 
F'_1 &= (F_1 \cap E') \cup \{e_T \mid T \in \mathcal{T}_S, \ \mbox{$|E(T) \cap F_1| = 1$} \}.  
\end{align*}
It is obvious that $(F'_0, F'_1)$ is a partition of $\delta_{G'}(S')$ 
and $b'(S') + |F'_1| \equiv b(S) + |F_1| \equiv 1 \pmod{2}$.  

To show that $\sum_{e \in F'_0} x'(e) + \sum_{e \in F'_1} (1-x'(e)) < 1$, 
we evaluate $x'(e_T)$ or $1-x'(e_T)$ for each $T \in \mathcal{T}_S$.  
Let $T \in \mathcal{T}_S$ be a triangle such that $E(T) = \{\alpha, \beta, \gamma\}$ and $E(T) \cap \delta_G(S) = \{\alpha, \beta\}$. 
Then, we obtain the following by the definition of $q^*(T)$.  
\begin{itemize}
\item
If $T \in \mathcal{T}_S$ and $\alpha, \beta \in F_0$, then 
$x(\alpha) + x(\beta) - 2 q^*(T) = x'(e_T)$. 
\item
If $T \in \mathcal{T}_S$ and $\alpha, \beta \in F_1$, then 
$(1-x(\alpha)) + (1-x(\beta)) - 2 q^*(T) = x'(e_T)$. 
\item
If $T \in \mathcal{T}_S$, $\alpha \in F_0$, and $\beta \in F_1$, then 
$x(\alpha) + (1-x(\beta)) - 2 q^*(T) = y_\emptyset + y_{\gamma} + y_{\alpha \beta} = 1-x'(e_T)$. 
\item
If $T \in \mathcal{T}_S$, $\beta \in F_0$, and $\alpha \in F_1$, then 
$(1-x(\alpha)) + x(\beta) - 2 q^*(T) = y_\emptyset + y_{\gamma} + y_{\alpha \beta} = 1-x'(e_T)$. 
\end{itemize}
With these observations, we obtain
\begin{align*}
 \sum_{e \in F'_0} x'(e) + \sum_{e \in F'_1} (1-x'(e))  
 = \sum_{e \in F_0} x(e) + \sum_{e \in F_1} (1- x(e)) - \sum_{T \in \mathcal{T}_S} 2 q^*(T) < 1, 
\end{align*}
which shows the ``only if'' part.

We next show the ``if'' part. 
For edge sets $F'_0, F'_1 \subseteq E'$, 
we denote $g(F'_0, F'_1) = \sum_{e \in F'_0} x'(e) + \sum_{e \in F'_1} (1-x'(e))$ to simplify the notation.  
Let $(S', F'_0, F'_1)$ be a minimizer of $g(F'_0, F'_1)$ subject to 
$(F'_0, F'_1)$ is a partition of $\delta_{G'}(S')$ and $b'(S') + |F'_1|$ is odd. 
Among minimizers, we choose $(S', F'_0, F'_1)$ so that $F'_0 \cup F'_1$ is inclusion-wise minimal. 
To derive a contradiction, assume that $g(F'_0, F'_1) < 1$. 
We show the following claim.

\begin{claim}
\label{clm:insidetriangle}
Let $T \in \mathcal{T}$ be a triangle as shown in Figure~\ref{fig:102} and 
denote $\hat{F}_0 = F'_0 \cap E'_T$ and $\hat{F}_1 = F'_1 \cap E'_T$. 
Then, we obtain the following. 
\begin{enumerate}
\item[(i)]
If $v_1, v_2, v_3 \not\in S'$, 
then $r_T \not\in S'$. 
\item[(ii)]
If $v_1, v_2, v_3 \in S'$, 
then $r_T \in S'$. 
\item[(iii)]
If $v_1 \in S'$, $v_2, v_3 \not\in S'$, and $|\hat{F}_1|$ is even, then
$g(\hat{F}_0, \hat{F}_1) = x'(e_1) = x(\alpha) + x(\gamma) - 2y_{\alpha \gamma}$. 
\item[(iv)]
If $v_1 \in S'$, $v_2, v_3 \not\in S'$, and $|\hat{F}_1|$ is odd, then
$g(\hat{F}_0, \hat{F}_1) = 1-x'(e_1) = y_{\emptyset} + y_{\beta} + y_{\alpha \gamma}$. 
\end{enumerate}
\end{claim}

\noindent
\textit{Proof of Claim\ref{clm:insidetriangle}. }
(i) Assume that $v_1, v_2, v_3 \not\in S'$ and $r_T \in S'$, 
which implies that $\hat{F}_0 \cup \hat{F}_1 = \{e_1, e_2, e_3\}$. 
Then, we derive a contradiction by the following case analysis and by the symmetry of $e_1, e_2$, and $e_3$.  
\begin{itemize}
\item
If $\hat{F}_0 = \{e_1, e_2\}$ and $\hat{F}_1 = \{e_3\}$, then 
\begin{align*}
g({F}'_0, {F}'_1) 
&\ge  g(\hat{F}_0, \hat{F}_1) \\
&= (x(\alpha)+x(\gamma) - 2y_{\alpha \gamma}) + (x(\alpha)+x(\beta)-2y_{\alpha \beta} ) 
+ (1-x(\beta)-x(\gamma)+2y_{\beta \gamma} ) \\
& = 1 + 2 y_\alpha + 2 y_{\beta \gamma} \ge 1,   
\end{align*}
which is a contradiction. 
\item
If $\hat{F}_0 = \emptyset$ and $\hat{F}_1 = \{e_1, e_2, e_3\}$, then
\begin{align*}
g({F}'_0, {F}'_1) 
&\ge g(\hat{F}_0, \hat{F}_1) \\
&= (1-x(\alpha)-x(\gamma) + 2y_{\alpha \gamma}) + (1-x(\alpha)-x(\beta)+2y_{\alpha \beta} ) + (1-x(\beta)-x(\gamma)+2y_{\beta \gamma} ) \\
&= 1 + 2(1-x(\alpha) - x(\beta) - x(\gamma) + y_{ \alpha \beta} + y_{\beta \gamma} + y_{ \alpha \gamma } ) \\
&= 1+2 y_\emptyset  
\ge 1,   
\end{align*}
which is a contradiction. 
\item
Suppose that $|\hat{F}_1|$ is even. 
Since $b'(S' \setminus \{r_T\}) + |F'_1 \setminus \delta_{G'}(r_T)|$ is odd
and $g(F'_0 \setminus \delta_{G'}(r_T), F'_1 \setminus \delta_{G'}(r_T)) \le g(F'_0, F'_1)$, 
$(S' \setminus \{r_T\}, F'_0 \setminus \delta_{G'}(r_T), F'_1 \setminus \delta_{G'}(r_T))$ is also a minimizer of $g$. 
This contradicts that  a minimizer $(S', F'_0, F'_1)$ is chosen so that $F'_0 \cup F'_1$ is inclusion-wise minimal. 
\end{itemize}

\medskip

(ii) Assume that $v_1, v_2, v_3 \in S'$ and $r_T \not\in S'$,  
which implies that $\hat{F}_0 \cup \hat{F}_1 = \{e_1, e_2, e_3\}$.  
Then, we derive a contradiction by the same argument as (i).

\medskip

(iii) 
Suppose that $v_1 \in S'$, $v_2, v_3 \not\in S'$, and $|\hat{F}_1|$ is even. 
Then, we have one of the following cases.  
\begin{itemize}
\item
If $\hat{F}_0 = \{e_1\}$ and $\hat{F}_1 = \emptyset$, then 
$$
g(\hat{F}_0, \hat{F}_1) = x'(e_1) = x(\alpha) + x(\gamma) - 2y_{\alpha \gamma}. 
$$
\item
If $\hat{F}_0 = \{e_2, e_3\}$ and $\hat{F}_1 = \emptyset$, then 
\begin{align*}
g(\hat{F}_0, \hat{F}_1) 
&= (x(\alpha) + x(\beta) - 2y_{\alpha \beta}) + (x(\beta) + x(\gamma) - 2y_{\beta \gamma})  \\
&= x(\alpha) + x(\gamma) + 2 ( x(\beta) - y_{\alpha \beta} - y_{\beta \gamma})  \\
&\ge x(\alpha)+x(\gamma) 
\ge x(\alpha) + x(\gamma) - 2y_{\alpha \gamma}. 
\end{align*}
\item
If $\hat{F}_0 = \emptyset$ and $\hat{F}_1 = \{e_2, e_3\}$, then
\begin{align*}
g(\hat{F}_0, \hat{F}_1) 
&= (1 - x(\alpha) - x(\beta) + 2y_{\alpha \beta}) + (1 - x(\beta) - x(\gamma) + 2y_{\beta \gamma}) \\
&= x(\alpha) + x(\gamma) - 2y_{\alpha \gamma}  + 2 (1-x(\alpha)-x(\beta)-x(\gamma)+ y_{\alpha \beta} + y_{\beta \gamma} + y_{\alpha \gamma}) \\
&\ge x(\alpha) + x(\gamma) - 2y_{\alpha \gamma}. 
\end{align*} 
\end{itemize}
Since $(S', F'_0, F'_1)$ is a minimizer of $g(F'_0, F'_1)$, 
$g(\hat{F}_0, \hat{F}_1) = x'(e_1) = x(\alpha) + x(\gamma) - 2y_{\alpha \gamma}$. 

\medskip

(iv) Suppose that $v_1 \in S'$, $v_2, v_3 \not\in S'$, and $|\hat{F}_1|$ is odd. 
Then, we have one of the following cases by changing the labels of $e_2$ and $e_3$ if necessary. 
\begin{itemize}
\item
If $\hat{F}_0 = \emptyset$ and $\hat{F}_1 = \{e_1\}$, then
$$
g(\hat{F}_0, \hat{F}_1) = 1-x'(e_1) = 1 - x(\alpha) - x(\gamma) + 2 y_{\alpha \gamma}. 
$$
\item
If $\hat{F}_0 = \{e_2\}$ and $\hat{F}_1 = \{e_3\}$, then
\begin{align*}
g(\hat{F}_0, \hat{F}_1) 
&= (x(\alpha) + x(\beta) - 2y_{\alpha \beta}) + (1 - x(\beta) - x(\gamma) + 2y_{\beta \gamma}) \\
&\ge 1-x(\alpha)-x(\gamma) + 2 ( x(\alpha) - y_{\alpha \beta}) \\ 
&\ge 1-x(\alpha)-x(\gamma) + 2y_{\alpha \gamma}.  
\end{align*}
\end{itemize}
Since $(S', F'_0, F'_1)$ is a minimizer of $g(F'_0, F'_1)$, 
$g(\hat{F}_0, \hat{F}_1) = 1-x'(e_1) = 1 - x(\alpha) - x(\gamma) + 2y_{\alpha \gamma} = y_{\emptyset} + y_{\beta} + y_{\alpha \gamma}$. 

\smallskip

\hspace{\fill}
(End of the Proof of Claim~\ref{clm:insidetriangle})

\bigskip

Note that each $T \in \mathcal{T}$ satisfies exactly one of (i)--(iv) of Claim~\ref{clm:insidetriangle}
by changing the labels of $v_1$, $v_2$, and $v_3$ if necessary.

In what follows, we construct $(S, F_0, F_1) \in \mathcal{F}$ for which 
 $(x, y)$ violates (\ref{eq:07}). 
We initialize $(S, F_0, F_1)$ as 
\begin{align*}
&S = S' \cap V, & &F_0 = F'_0 \cap E, & &F_1 = F'_1 \cap E, 
\end{align*}
and apply the following procedures for each triangle $T \in \mathcal{T}$. 
\begin{itemize}
\item
If $T$ satisfies the condition (i) or (ii) of Claim~\ref{clm:insidetriangle}, then we do nothing.  
\item
If $T$ satisfies the condition (iii) of Claim~\ref{clm:insidetriangle}, then we add $\alpha$ and $\gamma$ to $F_0$. 
\item
If $T$ satisfies the condition (iv) of Claim~\ref{clm:insidetriangle}, then we add $\alpha$ to $F_0$ and add $\gamma$ to $F_1$. 
\end{itemize}
Then, we obtain that  
$(F_0, F_1)$ is a partition of $\delta_G(S)$, 
$b(S) + |F_1| \equiv b'(S') + |F'_1| \equiv 1 \pmod{2}$, and
\begin{align*}
\sum_{e \in F_0} x(e) + \sum_{e \in F_1} (1- x(e)) - \sum_{T \in \mathcal{T}_S} 2 q^*(T) 
 =  \sum_{e \in F'_0} x'(e) + \sum_{e \in F'_1} (1-x'(e)) < 1
\end{align*}
by Claim~\ref{clm:insidetriangle}. 
This shows that $(x, y)$ violates (\ref{eq:07}) for $(S, F_0, F_1) \in \mathcal{F}$, 
which completes the proof of ``if'' part.  
\end{proof}

Since the proof of Lemma~\ref{lem:cutequiv} is constructive, 
given $S' \subseteq V'$ and $F'_0, F'_1 \subseteq E'$ such that 
$(F'_0, F'_1)$ is a partition of $\delta_{G'}(S')$, $b'(S') + |F'_1|$ is odd, and $\sum_{e \in F'_0} x'(e) + \sum_{e \in F'_1} (1-x'(e)) < 1$, 
we can construct 
$(S, F_0, F_1) \in \mathcal{F}$ for which 
 $(x, y)$ violates (\ref{eq:07}) in polynomial time. 
By combining this with Theorem~\ref{lem:oddmincut}, 
it holds that 
the separation problem for $P$ can be solved in polynomial time. 
Therefore, 
the ellipsoid method can maximize a linear function on $P$ in polynomial time (see e.g.~\cite{GLSbook}), 
and hence we can maximize $\sum_{e \in E} w(e) x(e)$ subject to $x \in {\rm proj}_E (P)$. 
By perturbing the objective function if necessary, 
we can obtain a maximizer $x^*$ that is an extreme point of ${\rm proj}_E (P)$. 
Since each extreme point of ${\rm proj}_E (P)$ corresponds to a $\mathcal{T}$-free $b$-factor by 
Theorem~\ref{thm:main01}, $x^*$ is a characteristic vector of a maximum weight $\mathcal{T}$-free $b$-factor. 
This completes the proof of Theorem~\ref{thm:algo}.

\section{Concluding Remarks}
\label{sec:conclusion}

This paper gives a first polynomial-time algorithm for the 
weighted $\mathcal{T}$-free $b$-matching problem 
where $\mathcal{T}$ is a set of edge-disjoint triangles.
A key ingredient is an extended formulation of the $\mathcal{T}$-free $b$-factor polytope
with exponentially many inequalities. 
As we mentioned in Section~\ref{sec:ef}, 
it is rare that the first polynomial-time algorithm was designed
with the aid of an extended formulation.
This approach has a potential to be used for 
other combinatorial optimization problems for which no polynomial-time algorithm is known. 

Some interesting problems remain open. 
Since the algorithm proposed in this paper relies on the ellipsoid method, it is natural to ask whether 
we can design a combinatorial polynomial-time algorithm. 
It is also open whether our approach can be applied
to the weighted $C_{\le 4}$-free $b$-matching problem in general graphs 
under the assumption that the forbidden cycles are edge-disjoint and 
the weight is vertex-induced on every square. 
In addition, the weighted $C_{\le 3}$-free $2$-matching problem and 
the $C_{\le 4}$-free $2$-matching problem are big open problems in this area.




\bibliographystyle{plainurl}

\bibliography{factor}

\appendix


\section{Proof of Claim~\ref{clm:key01}}
\label{sec:prf50}

By symmetry, it suffices to consider $(G_1, b_1, \mathcal{T}_1)$. 
Since the tightness of (\ref{eq:100}) for $(S^*, F^*_0, F^*_1)$ implies that $x_1(\delta_{G_1}(r)) = 1$, 
we can easily see that $(x_1, y_1)$ satisfies (\ref{eq:01}), (\ref{eq:02}), (\ref{eq:triangle})--(\ref{eq:06}). 
In what follows, we consider (\ref{eq:100}) for $(x_1, y_1)$ in $(G_1, b_1, \mathcal{T}_1)$. 
For edge sets $F'_0, F'_1 \subseteq E_1$, 
we denote $g(F'_0, F'_1) = \sum_{e \in F'_0} x_1(e) + \sum_{e \in F'_1} (1-x_1(e))$ to simplify the notation.  
For $(S', F'_0, F'_1) \in \mathcal{F}_1$, 
let $h(S', F'_0, F'_1)$ denote the left-hand side of (\ref{eq:100}). 
To derive a contradiction,  
let $(S', F'_0, F'_1) \in \mathcal{F}_1$ be a minimizer of $h(S', F'_0, F'_1)$ and 
assume that $h(S', F'_0, F'_1) < 1$.
By changing the roles of $S'$ and $V' \setminus S'$ if necessary, we may assume that $r \not\in S'$.

For $T \in \mathcal{T}^+_{S^*}$, 
let $v_1, v_2, v_3, \alpha, \beta$, and $\gamma$ be as in Figures~\ref{fig:105}--\ref{fig:108}. 
Let $G'_T = (V'_T, E'_T)$ be the subgraph of $G_1$ corresponding to $T$, that is, 
the subgraph induced by $\{r, p_1, p_2, v_2, v_3\}$ (Figure~\ref{fig:105}), 
$\{r, p_3, v_1 \}$ (Figure~\ref{fig:106}), $\{r, p_1, p_2, p_3, v_2, v_3\}$ (Figure~\ref{fig:107}), or 
$\{r, p_4, v_1 \}$ (Figure~\ref{fig:108}). 
Let $\hat{S} = S' \cap (V'_T \setminus \{v_1, v_2, v_3\})$, $\hat{F}_0 = F'_0 \cap E'_T$, and $\hat{F}_1 = F'_1 \cap E'_T$. 

We show the following properties (P1)--(P9) in Section~\ref{sec:P1}, and 
show that $(x_1, y_1)$ satisfies (\ref{eq:100}) by using these properties in Section~\ref{sec:Proof100}. 

\begin{description}
\item[(P1)]
If $T$ is of type (A) or (B) and $v_2, v_3 \not\in S'$, then $b_1(\hat{S}) + |\hat F_1|$ is even. 
\item[(P2)]
If $T$ is of type (A), $v_2, v_3 \in S'$, and $b_1(\hat{S}) + |\hat F_1|$ is even, then 
$g(\hat{F}_0, \hat{F}_1) \ge \min \{x(\alpha)+x(\gamma), 2-x(\alpha)-x(\gamma) - 2 y_\beta \}$. 
\item[(P3)]
If $T$ is of type (B), $v_2, v_3 \in S'$, and $b_1(\hat{S}) + |\hat F_1|$ is even, then 
$g(\hat{F}_0, \hat{F}_1) \ge y_{\alpha} + y_{\gamma} + y_{\alpha \beta} + y_{\beta \gamma}$. 
\item[(P4)]
If $T$ is of type (A) or (B), $v_2, v_3 \in S'$, and $b_1(\hat{S}) + |\hat F_1|$ is odd, then 
$g(\hat{F}_0, \hat{F}_1) \ge y_\emptyset + y_\beta + y_{\alpha \gamma}$. 
\item[(P5)]
If $T$ is of type (A) or (B), $v_2 \in S'$, $v_3 \not\in S'$, and $b_1(\hat{S}) + |\hat F_1|$ is even, then 
$g(\hat{F}_0, \hat{F}_1) \ge \min \{ x(\alpha)+x(\beta), 2 - x(\alpha) - x(\beta) - 2 y_{\gamma}\}$.  
\item[(P6)]
If $T$ is of type (A) or (B), $v_2 \in S'$, $v_3 \not\in S'$, and $b_1(\hat{S}) + |\hat F_1|$ is odd, then 
$g(\hat{F}_0, \hat{F}_1) \ge y_\emptyset + y_\gamma + y_{\alpha \beta}$.  
\item[(P7)]
If $T$ is of type (A') or type (B') and $v_1 \not\in S'$, then $b_1(\hat{S}) + |\hat F_1|$ is even. 
\item[(P8)]
If $T$ is of type (A') or type (B'), $v_1 \in S'$, and $b_1(\hat{S}) + |\hat F_1|$ is even, then 
$g(\hat{F}_0, \hat{F}_1) = \min \{x(\alpha)+x(\gamma), 2-x(\alpha)-x(\gamma) - 2 y_\beta \}$. 
\item[(P9)]
If $T$ is of type (A') or type (B'), $v_1 \in S'$, and $b_1(\hat{S}) + |\hat F_1|$ is odd, then 
$g(\hat{F}_0, \hat{F}_1) = y_\emptyset + y_\beta + y_{\alpha \gamma}$. 
\end{description}

Note that each $T \in \mathcal{T}^+_{S^*}$ satisfies exactly one of (P1)--(P9) 
by changing the labels of $v_2$ and $v_3$ if necessary.


\subsection{Proofs of (P1)--(P9)}
\label{sec:P1}

\subsubsection{When $T$ is of type (A)}
\label{sec:subcase(A)}

We first consider the case when $T$ is of type (A). 

\paragraph{Proof of (P1)}
Suppose that $T$ is of type (A) and $v_2, v_3 \not\in S'$. 
If $b_1(\hat S) + |\hat F_1|$ is odd, then 
either $p_1 \in \hat S$ and $|\hat F_1 \cap \delta_{G_1}(p_1)|$ is even or 
$p_2 \in \hat S$ and $|\hat F_1 \cap \delta_{G_1}(p_2)|$ is even. 
In the former case,
$h(S', F'_0, F'_1) \ge \min \{x_1(e_1) + x_1(e_5), 2 - x_1 (e_1)  - x_1(e_5) \} = 1$, which is a contradiction. 
The same argument can be applied to the latter case. 
Therefore, $b_1(\hat{S}) + |\hat F_1|$ is even. 

\paragraph{Proof of (P2)}
Suppose that $T$ is of type (A), $v_2, v_3 \in S'$, and $b_1(\hat{S}) + |\hat F_1|$ is even. 
If $p_1 \not\in S'$, then we define $(S'', F''_0, F''_1) \in \mathcal{F}_1$ as 
$(S'', F''_0, F''_1) = (S' \cup \{p_1\}, F'_0 \setminus \{e_5\}, F'_1\cup \{e_1\})$ if $e_5 \in F'_0$ and 
$(S'', F''_0, F''_1) = (S' \cup \{p_1\}, F'_0 \cup \{e_1\}, F'_1 \setminus \{e_5\})$ if $e_5 \in F'_1$. 
Since $h(S'', F''_0, F''_1) = h(S', F'_0, F'_1)$ holds, by replacing $(S', F'_0, F'_1)$ with $(S'', F''_0, F''_1)$, 
we may assume that $p_1 \in S'$. 
Similarly, we may assume that $p_2 \in S'$, which implies that
$\hat S = \{p_1, p_2\}$, $\hat F_0 \cup \hat F_1 = \{e_1, e_2, e_3, e_4\}$, and $|\hat F_1|$ is even. 
Then, 
$g(\hat F_0, \hat F_1) \ge \min\{ x(\alpha) + x(\gamma), 2-x(\alpha)-x(\gamma) - 2 y_\beta \}$ 
by the following case analysis. 

\begin{itemize}
\item
If $\hat F_1 = \emptyset$, then 
$g(\hat F_0, \hat F_1) = x_1(e_1) + x_1(e_2) + x_1(e_3) + x_1(e_4) = 2-x(\alpha)-x(\gamma) - 2 y_\beta$. 

\item
If $|\hat F_1| \ge 2$, then 
$g(\hat F_0, \hat F_1) \ge 2 - (x_1(e_1) + x_1(e_2) + x_1(e_3) + x_1(e_4))  =  x(\alpha) + x(\gamma) + 2 y_\beta \ge x(\alpha) + x(\gamma)$. 
\end{itemize}

\paragraph{Proof of (P4)}
Suppose that $T$ is of type (A), $v_2, v_3 \in S'$, and $b_1(\hat{S}) + |\hat F_1|$ is odd. 
In the same way as (P2), we may assume that 
$\hat S = \{p_1, p_2\}$, $\hat F_0 \cup \hat F_1 = \{e_1, e_2, e_3, e_4\}$, and $|\hat F_1|$ is odd. 
Then, 
$g(\hat F_0, \hat F_1) \ge y_\emptyset + y_\beta + y_{\alpha \gamma}$ 
by the following case analysis and by the symmetry of $v_2$ and $v_3$. 

\begin{itemize}
\item
If $|\hat F_1| = 3$, then $g(\hat F_0, \hat F_1) \ge 3 - (x_1(e_1) + x_1(e_2) + x_1(e_3) + x_1(e_4)) \ge 1 \ge y_\emptyset + y_\beta + y_{\alpha \gamma}$. 

\item
If $\hat F_1 = \{e_1\}$, then 
$g(\hat F_0, \hat F_1) \ge (1 - x_1(e_1)) + x_1(e_2) \ge y_\emptyset + y_\beta + y_{\alpha \gamma}$. 

\item
If $\hat F_1 = \{e_3\}$, then 
$g(\hat F_0, \hat F_1) \ge 1-x_1(e_3) \ge y_\emptyset + y_\beta + y_{\alpha \gamma}$. 
\end{itemize}

\paragraph{Proof of (P5)}
Suppose that $T$ is of type (A), $v_2 \in S'$, $v_3 \not\in S'$, and $b_1(\hat{S}) + |\hat F_1|$ is even. 
In the same way as (P2), we may assume that $p_1 \in S'$. 
If $p_2 \in S'$, then $b_1(p_2) + |F'_1 \cap \delta_{G_1}(p_2)|$ is even by the same calculation as (P1). 
Therefore, we may assume that $p_2 \not\in S'$, since otherwise
we can replace $(S', F'_0, F'_1)$ with $(S' \setminus \{p_2\}, F'_0 \setminus \delta_{G_1}(p_2), F'_1\setminus \delta_{G_1}(p_2))$ 
without increasing the value of $h(S', F'_0, F'_1)$. 
That is, we may assume that $\hat S = \{p_1\}$, $\hat F_0 \cup \hat F_1 = \{e_1, e_3\}$, and $|\hat F_1|$ is odd. 
Then, 
\begin{align*}
g(\hat F_0, \hat F_1) 
&\ge \min \{(1 - x_1(e_1)) + x_1(e_3),  x_1(e_1) + (1 - x_1(e_3)) \} \\
&= \min \{ x(\alpha)+x(\beta), 2 - x(\alpha) - x(\beta) \} \\ 
&\ge \min \{ x(\alpha)+x(\beta), 2 - x(\alpha) - x(\beta) - 2 y_{\gamma}\}. 
\end{align*}

\paragraph{Proof of (P6)}
Suppose that $T$ is of type (A), $v_2 \in S'$, $v_3 \not\in S'$, and $b_1(\hat{S}) + |\hat F_1|$ is odd. 
In the same way as (P5), 
we may assume that $\hat S = \{p_1\}$, $\hat F_0 \cup \hat F_1 = \{e_1, e_3\}$, and $|\hat F_1|$ is even. 
Then, 
\begin{align*}
g(\hat F_0, \hat F_1) 
&\ge \min \{ x_1(e_1) + x_1(e_3),  2 - x_1(e_1) - x_1(e_3) \} \\
&= \min \{ y_\emptyset + y_\gamma + y_{\alpha \beta}, 2 - (y_\emptyset + y_\gamma + y_{\alpha \beta})\} \\
&=  y_\emptyset + y_\gamma + y_{\alpha \beta}. 
\end{align*}

\subsubsection{When $T$ is of type (A')}
\label{sec:subcase(A')}

Second, we consider the case when $T$ is of type (A'). 

\paragraph{Proof of (P7)}
Suppose that $T$ is of type (A') and $v_1 \not\in S'$. 
If $b_1(\hat{S}) + |\hat F_1|$ is odd, then $\hat S =\{p_3\}$ and $|\hat F_1|$ is odd. 
This shows that $h(S', F'_0, F'_1) \ge g(\hat F_0, \hat F_1)  \ge 1$ by the following case analysis, which is a contradiction. 
\begin{itemize}
\item
If $\hat F_1 = \{e_1\}$, then $g(\hat F_0, \hat F_1) \ge (1-x_1(e_1)) + x_1(e_2) + x_1(e_9) \ge 1$. 
The same argument can be applied to the case of $\hat F_1 = \{e_2\}$ by the symmetry of $\alpha$ and $\gamma$. 

\item
If $\hat F_1 = \{e_i\}$ for some $i \in \{3, 4, 8\}$, then $g(\hat F_0, \hat F_1) \ge (1 - x_1(e_i)) + x_1(e_9) \ge 1$. 

\item
If $\hat F_1 = \{e_9\}$, then $g(\hat F_0, \hat F_1) 
= 1 + 2 y_\emptyset \ge 1$. 

\item
If $|\hat F_1| \ge 3$, then $g(\hat F_0, \hat F_1) \ge 3 - (x_1(e_1) + x_1(e_2) + x_1(e_3) + x_1(e_4) + x_1(e_8) + x_1(e_9)) \ge 1$. 
\end{itemize}
Therefore, $b_1(\hat S) + |\hat F_1|$ is even.

\paragraph{Proof of (P8)}
Suppose that $T$ is of type (A'), $v_1 \in S'$, and $b_1(\hat{S}) + |\hat F_1|$ is even. 
Then, 
$g(\hat F_0, \hat F_1) \ge \min\{ x(\alpha) + x(\gamma), 2-x(\alpha)-x(\gamma) - 2 y_\beta \}$ 
by the following case analysis. 
\begin{itemize}
\item
If $\hat F_0 = \{e_8, e_9\}$ and $\hat F_1 = \emptyset$, then 
$g(\hat F_0, \hat F_1) = x_1(e_8) + x_1(e_9) = x(\alpha) + x(\gamma)$. 

\item
If $\hat F_0 = \emptyset$ and $\hat F_1 = \{e_8, e_9\}$, then 
$g(\hat F_0, \hat F_1) = (1 - x_1(e_8)) + (1-x_1(e_9)) = 2-x(\alpha)-x(\gamma)  \ge 2-x(\alpha)-x(\gamma) - 2 y_\beta$. 

\item
If $\hat F_0 \cup \hat F_1 = \{e_1, e_2, e_3, e_4\}$, then 
$g(\hat F_0, \hat F_1) \ge \min\{ x(\alpha) + x(\gamma), 2-x(\alpha)-x(\gamma) - 2 y_\beta \}$ 
by the same calculation as (P2) in Section~\ref{sec:subcase(A)}. 
\end{itemize}

\paragraph{Proof of (P9)}
Suppose that $T$ is of type (A'), $v_1 \in S'$, and $b_1(\hat{S}) + |\hat F_1|$ is odd. 
Then, 
$g(\hat F_0, \hat F_1) \ge y_\emptyset + y_\beta + y_{\alpha \gamma}$ 
by the following case analysis. 
\begin{itemize}
\item
If $\hat F_0 = \{e_8 \}$ and $\hat F_1 = \{e_9 \}$, then 
$g(\hat F_0, \hat F_1) = x_1(e_8) + (1-x_1(e_9)) = y_\emptyset + y_\beta + y_{\alpha \gamma}$. 

\item
If $\hat F_0 = \{e_9 \}$ and $\hat F_1 = \{e_8 \}$, then 
$g(\hat F_0, \hat F_1) = (1 - x_1(e_8)) + x_1(e_9) \ge 1 \ge y_\emptyset + y_\beta + y_{\alpha \gamma}$. 

\item
If $\hat F_0 \cup \hat F_1 = \{e_1, e_2, e_3, e_4\}$, then 
$g(\hat F_0, \hat F_1) \ge y_\emptyset + y_\beta + y_{\alpha \gamma}$ 
by the same calculation as (P4) in Section~\ref{sec:subcase(A)}. 
\end{itemize}

\subsubsection{When $T$ is of type (B)}
\label{sec:subcase(B)}

Third, we consider the case when $T$ is of type (B). 
Let $G^+=(V^+, E^+)$ be the graph obtained from 
$G'_T = (V'_T, E'_T)$ in Figure~\ref{fig:107} by adding a new vertex $r^*$,
edges $e_{11}= r r^*$,  $e_{12}= v_2 r^*$, $e_{13}= v_3 r^*$, 
and self-loops $e_{14}, e_{15}, e_{16}$ that are incident to $v_2$, $v_3$, 
and $r^*$, respectively (Figure~\ref{fig:109}). 
We define $b_T: V^+ \to \mathbf{Z}_{\ge 0}$ as 
$b_T(v) = 1$ for $v \in \{r, p_1, p_2, p_3\}$ and $b_T(v) = 2$ for $v \in \{r^*, v_2, v_3\}$. 
We also define $x_T: E^+ \to \mathbf{Z}_{\ge 0}$ as 
$x_T(e) = x_1(e)$ for $e \in E'_T$ and
\begin{align*}
&x_T(e_{11}) = y_{\alpha} + y_{\gamma} + y_{\alpha \beta} + y_{\beta \gamma}, &  
&x_T(e_{12}) = y_{\alpha} + y_{\beta} + y_{\alpha \gamma} + y_{\beta \gamma}, &  
&x_T(e_{13}) = y_{\beta} + y_{\gamma} + y_{\alpha \beta} + y_{\alpha \gamma}, \\
&x_T(e_{14}) = y_{\emptyset} + y_{\gamma} , &  
&x_T(e_{15}) = y_{\emptyset} + y_{\alpha}, &  
&x_T(e_{16}) = y_{\emptyset}. 
\end{align*}
For $J \in \mathcal{E}_T$, define $b_T$-factors $M_J$ in $G^+$ as follows: 
\begin{align*}
& M_{\emptyset} = \{e_1, e_7, e_{14}, e_{15}, e_{16} \}, &
& M_{\alpha} = \{e_4, e_8, e_{11}, e_{12}, e_{15} \}, &
& M_{\beta} = \{e_1, e_8, e_9, e_{12}, e_{13} \}, \\
& M_{\gamma} = \{e_3, e_9, e_{11}, e_{13}, e_{14} \}, &
& M_{\alpha \beta} = \{e_5, e_8, e_9, e_{11}, e_{13} \}, &
& M_{\alpha \gamma} = \{e_2, e_8, e_9, e_{12}, e_{13} \}, \\
& M_{\beta \gamma} = \{e_6, e_8, e_9, e_{11}, e_{12}\}. & & & &
\end{align*}
Then,  we obtain 
$\sum_{J \in \mathcal{E}_T} y_1(J) =1$ and
$\sum_{J \in \mathcal{E}_T} y_1(J) x_{M_J}  = x_T$,
where $x_{M_J} \in \mathbf{R}^{E^+}$ is the characteristic vector of $M_J$.  
This shows that $x_T$ is in the $b_T$-factor polytope in $G^+$. 
Therefore, $x_T$ satisfies (\ref{eq:03}) with respect to $G^+$ and $b_T$. 
By using this fact, we show (P1), (P3), (P4), and (P6). 

\begin{figure}
\begin{center}
\includegraphics[clip,width=5.0cm]{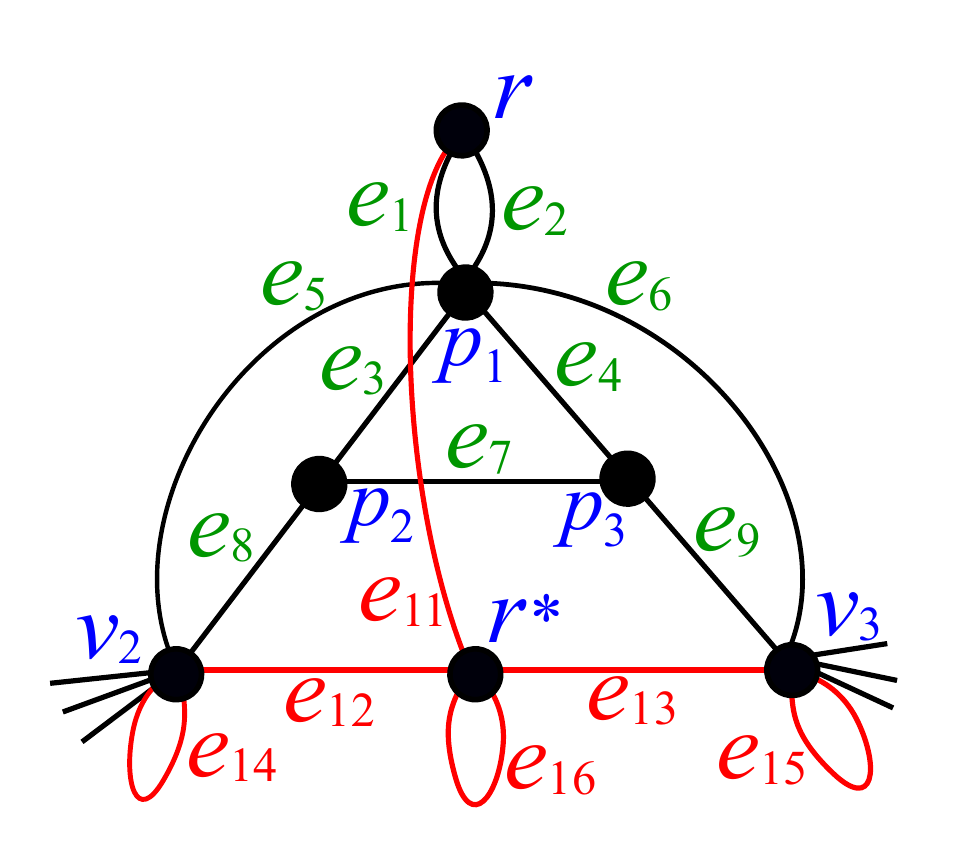}
\caption{Construction of $G^+$}
\label{fig:109}
\end{center}
\end{figure}

\paragraph{Proof of (P1)}
Suppose that $T$ is of type (B) and $v_2, v_3 \not\in S'$. 
If $b_1(\hat S) + |\hat F_1|$ is odd, then $b_T(\hat S) + |\hat F_1|$ is also odd. 
Since $x_T$ satisfies (\ref{eq:03}) with respect to $G^+$ and $b_T$, 
we obtain $g(\hat{F}_0, \hat{F}_1) \ge 1$. 
This shows that $h(S', F'_0, F'_1) \ge 1$, which is a contradiction. 
Therefore, $b_1(\hat S) + |\hat F_1|$ is even. 

\paragraph{Proof of (P3)}
Suppose that $T$ is of type (B), $v_2, v_3 \in S'$, and $b_1(\hat{S}) + |\hat F_1|$ is even. 
Since $b_T(\hat S \cup \{r^*, v_2, v_3\}) + |\hat F_1 \cup \{e_{11}\}|$ is odd and $x_T$ satisfies (\ref{eq:03}), we obtain  
$g(\hat F_0, \hat F_1) + (1-x_T(e_{11})) \ge 1$. 
Therefore, 
$g(\hat F_0, \hat F_1) \ge x_T(e_{11}) = y_\alpha + y_\gamma + y_{\alpha \beta} + y_{\beta \gamma}$.

\paragraph{Proof of (P4)}
Suppose that $T$ is of type (B), $v_2, v_3 \in S'$, and $b_1(\hat{S}) + |\hat F_1|$ is odd. 
Since $b_T(\hat S \cup \{r^*, v_2, v_3\}) + |\hat F_1|$ is odd and $x_T$ satisfies (\ref{eq:03}), we obtain  
$g(\hat F_0, \hat F_1) + x_T(e_{11}) \ge 1$. 
Therefore, 
$g(\hat F_0, \hat F_1) \ge 1 - x_T(e_{11}) = y_\emptyset + y_\beta + y_{\alpha \gamma}$.  

\paragraph{Proof of (P6)}
Suppose that $T$ is of type (B), $v_2 \in S'$, $v_3 \not\in S'$, and $b_1(\hat{S}) + |\hat F_1|$ is odd. 
Since $b_T(\hat S \cup \{v_2 \}) + |\hat F_1|$ is odd and $x_T$ satisfies (\ref{eq:03}), we obtain  
$g(\hat F_0, \hat F_1) + x_T(e_{12}) \ge 1$. 
Therefore, 
$g(\hat F_0, \hat F_1) \ge 1 - x_T(e_{12}) = y_\emptyset + y_\gamma + y_{\alpha \beta}$.  

\bigskip

In what follows, we show (P5) by a case analysis.

\paragraph{Proof of (P5)}
Suppose that $T$ is of type (B), $v_2 \in S'$, $v_3 \not\in S'$, and $b_1(\hat{S}) + |\hat F_1|$ is even. 
If $\hat S \cap \{p_1, p_3\} \not= \emptyset$ and $p_2 \not\in \hat S$, then 
we can add $p_2$ to $S'$ without decreasing the value of $h(S', F'_0, F'_1)$. 
Therefore, we can show 
$g(\hat F_0, \hat F_1) \ge \{x(\alpha) + x(\beta), 2 - x(\alpha) - x(\beta) - 2 y_\gamma \}$ 
by the following case analysis. 
\begin{itemize}
\item
Suppose that $\hat S = \{p_1, p_2, p_3\}$, which implies that $\hat F_0 \cup \hat F_1 = \{e_1, e_2, e_6, e_9\}$ and $|\hat F_1|$ is odd. 
\begin{itemize}
\item
If $\hat F_1 = \{e_i\}$ for $i \in \{1, 2, 6\}$, then 
$g(\hat F_0, \hat F_1) \ge (1- x_1(e_i)) + x_1(e_9) \ge x(\alpha) + x(\beta)$. 
\item
If $\hat F_1 = \{e_9\}$, then 
$g(\hat F_0, \hat F_1) =  y_\alpha + y_\beta + y_{\alpha \gamma} + y_{\beta \gamma} + 2 y_\emptyset =  2 - x(\alpha) - x(\beta) - 2 y_\gamma$. 
\item
If $|\hat F_1| = 3$, then 
$g(\hat F_0, \hat F_1) \ge 3 - (x_1(e_1) + x_1(e_2) + x_1(e_6) + x_1(e_9)) \ge x(\alpha) + x(\beta)$. 
\end{itemize}

\item
Suppose that $\hat S = \{p_1, p_2\}$, which implies that $\hat F_0 \cup \hat F_1 = \{e_1, e_2, e_4, e_6, e_7\}$ and $|\hat F_1|$ is even. 
\begin{itemize}
\item
If $\hat F_1 = \emptyset$, then  
$g(\hat F_0, \hat F_1) = x_1(e_1) + x_1(e_2) + x_1(e_4) + x_1(e_6) + x_1(e_7) =  2 - x(\alpha) - x(\beta) - 2 y_\gamma$. 
\item
If $|\hat F_1| \ge 2$, then 
$g(\hat F_0, \hat F_1) \ge 2 - (x_1(e_1) + x_1(e_2) + x_1(e_4) + x_1(e_6) + x_1(e_7)) \ge x(\alpha) + x(\beta)$. 
\end{itemize}

\item
Suppose that $\hat S = \{p_2\}$, which implies that $\hat F_0 \cup \hat F_1 = \{e_3, e_5, e_7\}$ and $|\hat F_1|$ is odd.  
\begin{itemize}
\item
If $\hat F_1 = \{e_i\}$ for $i \in \{3, 7\}$, then 
$g(\hat F_0, \hat F_1) \ge (1- x_1(e_i)) + x_1(e_5) \ge x(\alpha) + x(\beta)$. 
\item
If $\hat F_1 = \{e_5\}$, then 
$g(\hat F_0, \hat F_1) \ge (1-x_1(e_5)) + x_1(e_7) \ge 2 - x(\alpha) - x(\beta) - 2 y_\gamma$. 
\item
If $\hat F_1 = \{e_3, e_5, e_7\}$, then
$g(\hat F_0, \hat F_1) = 3 - (x_1(e_3) + x_1(e_5) + x_1(e_7)) \ge x(\alpha) + x(\beta)$. 
\end{itemize}

\item
Suppose that $\hat S = \{p_2, p_3\}$, which implies that $\hat F_0 \cup \hat F_1 = \{e_3, e_4, e_5, e_9\}$ and $|\hat F_1|$ is even.  
\begin{itemize}
\item
If $\hat F_1 = \emptyset$, then 
$g(\hat F_0, \hat F_1) = x_1(e_3) + x_1(e_4) + x_1(e_5) + x_1(e_9) =  x(\alpha) + x(\beta) + 2 y_{\gamma} \ge x(\alpha) + x(\beta)$. 
\item
If $|\hat F_1| \ge 2$, then 
$g(\hat F_0, \hat F_1) \ge 2 - (x_1(e_3) + x_1(e_4) + x_1(e_5) + x_1(e_9)) =  2- x(\alpha) - x(\beta) - 2 y_{\gamma}$. 
\end{itemize}

\item 
If $\hat S = \emptyset$, then $\hat F_0 \cup \hat F_1 = \{e_5, e_8\}$ and $|\hat F_1|$ is even.  
Therefore, 
$g(\hat F_0, \hat F_1) \ge \min \{x_1(e_5) + x_1(e_8), 2-x_1(e_5)-x_1(e_8) \} 
\ge \min\{x(\alpha) + x(\beta), 2 - x(\alpha) - x(\beta) - 2 y_\gamma \}$.  
\end{itemize}


\subsubsection{When $T$ is of type (B')}
\label{sec:subcase(B')}

Finally, we consider the case when $T$ is of type (B'). 

\paragraph{Proof of (P7)}
Suppose that $T$ is of type (B') and $v_1 \not\in S'$. 
If $b_1(\hat{S}) + |\hat F_1|$ is odd, then 
$\hat S =\{p_4\}$ and 
$h(S', F'_0, F'_1) \ge \min \{x_1(e_1) + x_1(e_{10}), 2 - x_1(e_1) - x_1(e_{10})\} = 1$, 
which is a contradiction. 
Therefore, $b_1(\hat{S}) + |\hat F_1|$ is even. 


\paragraph{Proof of (P8)}
Suppose that $T$ is of type (B'), $v_1 \in S'$, and $b_1(\hat{S}) + |\hat F_1|$ is even. 
If $p_4 \not\in S'$, then we define $(S'', F''_0, F''_1) \in \mathcal{F}_1$ as 
$(S'', F''_0, F''_1) = (S' \cup \{p_4\}, F'_0 \setminus \{e_{10}\}, F'_1\cup \{e_1\})$ if $e_{10} \in F'_0$ and 
$(S'', F''_0, F''_1) = (S' \cup \{p_4\}, F'_0 \cup \{e_1\}, F'_1 \setminus \{e_{10}\})$ if $e_{10} \in F'_1$. 
Since $h(S'', F''_0, F''_1) = h(S', F'_0, F'_1)$, by replacing $(S', F'_0, F'_1)$ with $(S'', F''_0, F''_1)$, 
we may assume that $p_4 \in S'$. 
Then, since $\hat F_0 \cup \hat F_1 = \{e_1, e_2\}$ and $|\hat F_1|$ is odd, we obtain
\begin{align*}
g(\hat F_0, \hat F_1) 
&\ge \min\{ (1-x_1(e_1)) + x_1(e_2), x_1(e_1) + (1-x_1(e_2)) \} \\
&\ge \min\{ x(\alpha) + x(\gamma), 2 - x(\alpha) - x(\gamma) - 2y_\beta \}.  
\end{align*}


\paragraph{Proof of (P9)}
Suppose that $T$ is of type (B'), $v_1 \in S'$, and $b_1(\hat{S}) + |\hat F_1|$ is odd. 
In the same way as (P8), 
we may assume that $\hat S = \{p_4\}$, $\hat F_0 \cup \hat F_1 = \{e_1, e_2\}$, and $|\hat F_1|$ is even. 
Then, 
\begin{align*}
g(\hat F_0, \hat F_1) 
&\ge \min\{ x_1(e_1) + x_1(e_2), (1-x_1(e_1)) + (1-x_1(e_2)) \} \\
&= \min\{ y_\emptyset + y_\beta + y_{\alpha \gamma}, 2 - (y_\emptyset + y_\beta + y_{\alpha \gamma})\} \\
&=  y_\emptyset + y_\beta + y_{\alpha \gamma}.  
\end{align*}

\subsection{Condition (\ref{eq:100})}
\label{sec:Proof100}

Recall that $r \not\in S'$ is assumed and note that $x_1(\delta_{G_1}(r)) = 1$. 
Let $\mathcal{T}_{(P3)} \subseteq \mathcal{T}^+_{S^*}$ be the set of triangles satisfying the conditions in (P3), 
i.e., the set of triangles of type (B) such that $v_2, v_3 \in S'$ and $b_1(\hat{S}) + |\hat F_1|$ is even. 
Since $y_{\alpha} + y_{\gamma} + y_{\alpha \beta} + y_{\beta \gamma} = 1 - x_1(e^{T}_1) - x_1(e^{T}_2)$ holds for each triangle $T \in \mathcal{T}^+_{S^*}$ of type (B), 
if there exist two triangles $T, T' \in \mathcal{T}_{(P3)}$, then 
$h(S', F'_0,  F'_1) \ge (1- x_1(e^{T}_1) - x_1(e^{T}_2)) + (1- x_1(e^{T'}_1) - x_1(e^{T'}_2)) \ge 2 - x_1(\delta_{G_1}(r)) = 1$, which is a contradiction. 
Similarly, if there exists a triangle $T \in \mathcal{T}_{(P3)}$ and an edge $e \in (\delta_{G_1}(r) \setminus E'_{T}) \cap F'_1$, 
then $h(S', F'_0,  F'_1) \ge (1- x_1(e^{T}_1) - x_1(e^{T}_2)) + (1- x_1(e)) \ge 2 - x_1(\delta_{G_1}(r)) = 1$, which is a contradiction. 
Therefore,  either $\mathcal{T}_{(P3)} = \emptyset$ holds or 
$\mathcal{T}_{(P3)}$ consists of exactly one triangle, say $T$, and $(\delta_{G_1}(r) \setminus E'_{T}) \cap F'_1 = \emptyset$.

Assume that $\mathcal{T}_{(P3)} = \{T\}$ and $(\delta_{G_1}(r) \setminus E'_{T}) \cap F'_1 = \emptyset$. 
Define $(S'', F''_0, F''_1) \in \mathcal{F}_1$ as 
\begin{align*}
& S'' = S' \cup V'_{T}, & &  F''_0 = (F'_0  \triangle \delta_{G_1}(r)) \setminus E'_{T}, &  & F''_1 = F'_1 \setminus E'_{T}, 
\end{align*}
where $\triangle$ denotes the symmetric difference. 
Note that $(F''_0, F''_1)$ is a partition of $\delta_{G_1} (S'')$, 
$b_1(S'') + |F''_1| = (b_1 (S') + b_1(\hat{S}) ) + (|F'_1| - |\hat F_1|) \equiv 1 \pmod{2}$, and 
$h(S', F'_0, F'_1) - h(S'', F''_0, F''_1) \ge (1 - x_1(e^{T}_1) - x_1(e^{T}_2)) - x_1( \delta_{G_1}(r) \setminus \{ x_1(e^{T}_1), x_1(e^{T}_2)\}) = 0$. 
By these observations, $(S'', F''_0, F''_1) \in \mathcal{F}_1$ is also a minimizer of $h$. 
This shows that $(V'' \setminus S'', F''_0, F''_1) \in \mathcal{F}_1$ is a minimizer of $h$ such that 
$r \in V'' \setminus S''$. 
Furthermore, if a triangle $T' \in \mathcal{T}^+_{S^*}$ satisfies the conditions in (P3) with respect to $(V'' \setminus S'', F''_0, F''_1)$, 
then $T'$ is a triangle of type (B) such that $v_2, v_3 \not\in S'$ and $b_1(\hat{S}) + |\hat F_1|$ is odd with respect to $(S', F'_0, F'_1)$, 
which contradicts (P1). 
Therefore, by replacing $(S', F'_0,  F'_1)$ with $(V'' \setminus S'', F''_0, F''_1)$, 
we may assume that $\mathcal{T}_{(P3)} = \emptyset$. 

In what follows, we construct $(S, F_0, F_1) \in \mathcal{F}$ for which 
 $(x, y)$ violates (\ref{eq:100}) to derive a contradiction. 
We initialize $(S, F_0, F_1)$ as 
\begin{align*}
& S = S' \cap V, & & F_0 = F'_0 \cap E, & &  F_1 = F'_1 \cap E,  
\end{align*}
and apply the following procedures for each triangle $T \in \mathcal{T}^+_{S^*}$. 
\begin{itemize}
\item
Suppose that $T$ satisfies the condition in (P1) or (P7). 
In this case, we do nothing. 
\item
Suppose that $T$ satisfies the condition in (P2) or (P8). 
If $g(\hat F_0, \hat F_1) \ge x(\alpha) + x(\gamma)$, then 
add $\alpha$ and $\gamma$ to $F_0$. 
Otherwise, since $g(\hat F_0, \hat F_1) \ge 2-x(\alpha)-x(\gamma) - 2 y_\beta$, 
add $\alpha$ and $\gamma$ to $F_1$.
\item
Suppose that $T$ satisfies the condition in (P4) or (P9). 
In this case, add $\alpha$ to $F_0$ and add $\gamma$ to $F_1$. 
\item
Suppose that $T$ satisfies the condition in (P5). 
If $g(\hat F_0, \hat F_1) \ge x(\alpha) + x(\beta)$, then 
add $\alpha$ and $\beta$ to $F_0$. 
Otherwise, since $g(\hat F_0, \hat F_1) \ge 2-x(\alpha)-x(\beta) - 2 y_\gamma$, 
add $\alpha$ and $\beta$ to $F_1$.
\item
Suppose that $T$ satisfies the condition in (P6). 
In this case, add $\alpha$ to $F_0$ and add $\beta$ to $F_1$. 
\end{itemize}
Note that exactly one of the above procedures is applied for each $T \in \mathcal{T}^+_{S^*}$, 
because $\mathcal{T}_{(P3)} = \emptyset$. 

Then, we see that $(S, F_0, F_1) \in \mathcal{F}$ holds and 
the left-hand side of (\ref{eq:100}) with respect to $(S, F_0, F_1)$ is 
at most $h(S', F'_0, F'_1)$ by (P1)--(P9). 
Since $h(S', F'_0, F'_1) < 1$ is assumed,  
 $(x, y)$ violates (\ref{eq:100}) for 
$(S, F_0, F_1) \in \mathcal{F}$, 
which is a contradiction.  \qed


\section{Proof of Claim~\ref{clm:key03}}
\label{sec:key03}

We can easily see that replacing $(M_1 \cup M_2)  \cap \{e^f \mid f \in \tilde{F}^*_0\}$ with $\{ f \in \tilde{F}^*_0 \mid {e}^f \in M_1 \cap M_2 \}$ does not 
affect the degrees of vertices in $V$. 
Since $M_1 \cup M_2$ contains exactly one of $\{e^f_u, e^f_v\}$ or ${e}^f_r (= {e}^f_{r'})$ for $f = uv \in \tilde{F}^*_1$, 
replacing $(M_1 \cup M_2)  \cap \{e^f_u, e^f_r, e^f_v \mid f = uv \in \tilde{F}^*_1\}$ with 
$\{ f \in \tilde{F}^*_1 \mid e^f_r \not\in M_1 \cap M_2 \}$ does not 
affect the degrees of vertices in $V$. 

For every $T \in \mathcal{T}^+_{S^*}$ of type (A) or (A'), 
since 
\begin{align*}
&|\varphi(M_1, M_2, T) \cap \{\alpha, \gamma\}| = |M_T \cap \{e_8, e_9\}|,  \\
&|\varphi(M_1, M_2, T) \cap \{\alpha, \beta\}| = |M_T \cap \{e_3, e_5\}|, \mbox{ and}   \\
&|\varphi(M_1, M_2, T) \cap \{\beta, \gamma\}| = |M_T \cap \{e_4, e_6\}|   
\end{align*}
hold by the definition of $\varphi(M_1, M_2, T)$, replacing $M_T$ with $\varphi(M_1, M_2, T)$ 
does not affect the degrees of vertices in $V$. 

Furthermore, for every $T \in \mathcal{T}^+_{S^*}$ of type (B) or (B'), 
since 
\begin{align*}
&|\varphi(M_1, M_2, T) \cap \{\alpha, \gamma\}| = |M_T \cap \{e_2, e_{10}\}|,  \\
&|\varphi(M_1, M_2, T) \cap \{\alpha, \beta\}| = |M_T \cap \{e_5, e_8\}|, \mbox{ and}   \\
&|\varphi(M_1, M_2, T) \cap \{\beta, \gamma\}| = |M_T \cap \{e_6, e_9\}|   
\end{align*}
hold by the definition of $\varphi(M_1, M_2, T)$, replacing $M_T$ with $\varphi(M_1, M_2, T)$ 
does not affect the degrees of vertices in $V$. 

Since $b(v) = b_1(v)$ for $v \in S^*$ and $b(v) = b_2(v)$ for $v \in V^* \setminus S^*$, 
this shows that 
$M_1 \oplus M_2$ forms a $b$-factor. 
Since $M_j$ is $\mathcal{T}_j$-free for $j \in \{1, 2\}$, 
$M_1 \oplus M_2$ is a $\mathcal{T}$-free $b$-factor. \qed


\section{Proof of Claim~\ref{clm:key04}}
\label{sec:key04}

By the definitions of $x_1, x_2$, and $M_1 \oplus M_2$, (\ref{eq:52}) and (\ref{eq:53}) show that
\begin{equation}
x(e) = \sum_{(M_1, M_2) \in \mathcal{M}}  \lambda_{(M_1, M_2)} x_{M_1 \oplus M_2}(e)  \label{eq:80}
\end{equation} 
for $e \in E \setminus \bigcup_{T \in \mathcal{T}^+_{S^*}} E(T)$.

Let $T \in \mathcal{T}^+_{S^*}$ be a triangle of type (A) for $(G_1, b_1, \mathcal{T}_1)$
and let $\alpha, \beta$, and $\gamma$ be as in Figures~\ref{fig:105} and~\ref{fig:106}. 
By the definition of $\varphi(M_1, M_2, T)$, we obtain
\begin{align*}
&\sum_{(M_1, M_2) \in \mathcal{M}}  \lambda_{(M_1, M_2)} x_{M_1 \oplus M_2}(\beta) \\
&\quad = \sum \{ \lambda_{(M_1, M_2)} \mid  \mbox{$\varphi(M_1, M_2, T) = \{\alpha, \beta\}, \{\beta, \gamma\}$, or $\{\beta \}$} \}  \\
&\quad = \sum \{ \lambda_{(M_1, M_2)} \mid  e_3 \in M_T \} + \sum \{ \lambda_{(M_1, M_2)} \mid  e_4 \in M_T \} + \sum \{ \lambda_{(M_1, M_2)} \mid  e_7 \in M_T \} \\
&\quad = x_1(e_3) + x_1(e_4) + x_2(e_7) \\
&\quad = y_{\alpha \beta} + y_{\beta \gamma} + y_\beta = x(\beta).  
\end{align*}
We also obtain
\begin{align*}
&\sum_{(M_1, M_2) \in \mathcal{M}}  \lambda_{(M_1, M_2)} x_{M_1 \oplus M_2}(\alpha)  \\
&\quad = \sum \{ \lambda_{(M_1, M_2)} \mid  \varphi(M_1, M_2, T) \not= \{\gamma \}, \{\beta, \gamma\}, \{\beta\} \}  \\
&\quad = 1 - \sum \{ \lambda_{(M_1, M_2)} \mid  e_1 \in M_T \} - \sum \{ \lambda_{(M_1, M_2)} \mid  e_4 \in M_T \} - \sum \{ \lambda_{(M_1, M_2)} \mid  e_7 \in M_T \} \\
&\quad = 1 - x_1(e_1) - x_1(e_4) - x_2(e_7) \\ 
&\quad = 1- y_{\emptyset} - y_{\gamma} - y_{\beta \gamma} - y_\beta = x(\alpha).   
\end{align*}
Since a similar equality holds for $\gamma$ by symmetry, (\ref{eq:80}) holds for $e \in \{\alpha, \beta, \gamma\}$. 
Since $T$ is a triangle of type (A') for $(G_1, b_1, \mathcal{T}_1)$ if and only if it is of type (A) for $(G_2, b_2, \mathcal{T}_2)$, 
the same argument can be applied when $T$ is a triangle of type (A') for $(G_1, b_1, \mathcal{T}_1)$.

Let $T \in \mathcal{T}^+_{S^*}$ be a triangle of type (B) for $(G_1, b_1, \mathcal{T}_1)$
and let $\alpha, \beta$, and $\gamma$ be as in Figures~\ref{fig:107} and~\ref{fig:108}. 
By the definition of $\varphi(M_1, M_2, T)$, we obtain
\begin{align*}
&\sum_{(M_1, M_2) \in \mathcal{M}}  \lambda_{(M_1, M_2)} x_{M_1 \oplus M_2}(\beta)  \\
&\quad = \sum \{ \lambda_{(M_1, M_2)} \mid \varphi(M_1, M_2, T) \not= \emptyset, \{\alpha \}, \{\gamma \}, \{\alpha, \gamma \} \}  \\
&\quad = 1 - \sum \{ \lambda_{(M_1, M_2)} \mid  e_2 \in M_T \} - \sum \{ \lambda_{(M_1, M_2)} \mid  e_3 \in M_T \}  \\ 
&\qquad \qquad - \sum \{ \lambda_{(M_1, M_2)} \mid  e_4 \in M_T \} - \sum \{ \lambda_{(M_1, M_2)} \mid  e_7 \in M_T \} \notag \\
&\quad = 1 - x_1(e_2) - x_1(e_3) - x_1(e_4) - x_1(e_7) \\
&\quad = 1 - y_{\alpha \gamma} - y_{\gamma} - y_{\alpha} - y_{\emptyset} = x(\beta).  
\end{align*}
We also obtain
\begin{align*}
&\sum_{(M_1, M_2) \in \mathcal{M}}  \lambda_{(M_1, M_2)} x_{M_1 \oplus M_2}(\alpha) \\
&\quad = \sum \{ \lambda_{(M_1, M_2)} \mid  \mbox{$\varphi(M_1, M_2, T) = \{\alpha \}, \{\alpha, \beta\}$, or $\{\alpha, \gamma \}$} \} \\
&\quad = \sum \{ \lambda_{(M_1, M_2)} \mid  e_2 \in M_T \} + \sum \{ \lambda_{(M_1, M_2)} \mid  e_4 \in M_T \} + \sum \{ \lambda_{(M_1, M_2)} \mid  e_5 \in M_T \} \\
&\quad = x_1(e_2) + x_1(e_4) + x_1(e_5) \\ 
&\quad = y_{\alpha \gamma} + y_{\alpha} + y_{\alpha \beta} = x(\alpha).  
\end{align*}
Since a similar equality holds for $\gamma$ by symmetry, (\ref{eq:80}) holds for $e \in \{\alpha, \beta, \gamma\}$. 
The same argument can be applied when $T$ is a triangle of type (B') for $(G_1, b_1, \mathcal{T}_1)$. 

Therefore, (\ref{eq:80}) holds for every $e \in E$, which complete the proof of the claim.  \qed

\end{document}